\documentclass[final,1p,times]{elsarticle}

\usepackage{amscd} 
\usepackage{mathrsfs}
\usepackage[english]{babel}
\usepackage{amsfonts,amsmath,amsxtra,amsthm,amssymb,latexsym}
\usepackage{verbatim}
\biboptions{square,comma,numbers}

\newtheorem{theorem}{Theorem}[section]
\newtheorem{lemma}[theorem]{Lemma}
\newtheorem{corollary}[theorem]{Corollary} 
\theoremstyle{definition}

\newtheorem{remark}[theorem]{Remark} 
\numberwithin{equation}{section}

\def\z*{\bar z}

\def\B{\mathscr B}

\def\uno{\mathsf 1}
\def\H{\mathsf H}
\def\fh{\mathfrak h}
\def\fb{\mathfrak b}

\def\X{\mathsf X}
\def\Y{\mathsf Y}

\def\D{\text{\rm dom}}
\def\dom{\text{\rm dom}}
\def\ran{\text{\rm ran}}
\def\supp{\text{\rm supp}}

\def\RE{\mathbb R}
\def\CO{{\mathbb C}}

\def\SL{{\rm SL}}
\def\DL{{\rm DL}}

\def\NA{\mathbb N}

\def\ph*{\phi_\star}

\def\be{\begin{equation}}
\def\ee{\end{equation}}
\def\min{{\rm min}}
\def\max{{\rm max}}

\def\n{{\|}}
\def\-{{\rm in}}
\def\+{{\rm ex}}
\journal{Journal de Math\'ematiques Pures et Appliqu\'ees}
\begin{document}
\begin{frontmatter}

\author[man]{Andrea Mantile}
\ead{andrea.mantile@univ-reims.fr}
\address[man]{Laboratoire de Math\'{e}matiques, Universit\'{e} de Reims -
FR3399 CNRS, Moulin de la Housse BP 1039, 51687 Reims, France}

\author[pos]{Andrea Posilicano\corref{cor}}
\ead{andrea.posilicano@uninsubria.it}
\cortext[cor]{Corresponding author}
\address[pos]{DiSAT - Sezione di Matematica, Universit\`a dell'Insubria, Via Valleggio 11, I-22100 Como, Italy}

\title
{Asymptotic Completeness and S-Matrix for Singular Perturbations}

\begin{abstract} 
We give a criterion of asymptotic completeness and provide a representation of the scattering matrix for the scattering couple $(A_{0},A)$, where $A_{0}$ and $A$ are semi-bounded self-adjoint operators in $L^{2}(M,{\mathcal B},m)$ such that the set $\{u\in\D(A_{0})\cap\D(A):A_{0}u=Au\}$ is dense. No sort of trace-class condition on resolvent differences is required. Applications to the case in which $A_{0}$ corresponds to the free Laplacian in $L^{2}({\mathbb R}^{n})$ and $A$ describes the Laplacian with self-adjoint boundary conditions on rough compact hypersurfaces are given.
\vskip 10pt 
\noindent
{\bf R\'esum\'e}
\vskip8pt\noindent
On fournit un crit\`{e}re pour la compl\'{e}tude asymptotique et une
repr\'{e}sentation de la matrice de la diffusion pour un syst\`{e}me de
diffusion $(  A_{0},A)  $, \'{e}tant $A_{0}$ et $A$ op\'{e}rateurs
autoadjointes demi-born\'{e}s dans $L^{2}(M,{\mathcal B},m)$
tels que l'ensemble  $\{u\in\D(A_{0})\cap\D(A):A_{0}u=Au\}$ est dense. Aucune
condition de trace sur les r\'{e}solvantes est requise. On consid\`{e}re des
applications aux cas o\`{u} $A_{0}$ est le Laplacien libre dans $L^{2}({\mathbb R}^{n})$ et $A$ d\'{e}crit le Laplacien avec conditions au
bord autoadjointes sur une hypersurface compacte et non r\'{e}guli\`{e}re.

\end{abstract}

\end{frontmatter}


\begin{section}{Introduction.}
Let $A_{0}:\dom(A_{0})\subseteq\H\to\H$ be a self-adjoint operator in the
Hilbert space $\H$. Another self-adjoint operator $A:\dom(A)\subseteq\H\to\H$ is said to
be a singular perturbation of $A_{0}$ if the set $N_{0}:=\{u\in
\dom(A_{0})\cap \dom(A)\, :\, A_{0}u=Au\}$ is dense in
$\H$ (see e.g. \cite{Kosh}, \cite{P04}); in typical situations $A_{0}$ and $A$ correspond to the same differential expression and they differ due to some boundary conditions imposed on a null subset.\par Since the subspace $N_{0}$ is closed with respect to the graph norm of $A_{0}$, the linear operator $S_{0}:=A_{0}|N_{0}$, obtained by restricting $A_{0}$ to $N_{0}$, is a densely defined closed symmetric operator and $A$ is one of its self-adjoint extensions. Therefore to find all singular perturbations of $A_{0}$ it suffices to pick out $\H$-dense subspaces $N\subsetneq \D(A_{0})$, closed with respect to the graph
norm of $A_{0}$, and then to look for the self-adjoint
extensions of $S=A_{0}|N$: for any of such a self-adjoint extensions $A\not=A_{0}$ one has $\D(S)=N\subseteq N_{0}=\{u\in
\dom(A_{0})\cap \dom(A)\, :\, A_{0}u=Au\}$ and so $N_{0}$ is dense. Since $\D(A_{0})$ is a Hilbert   space with respect to the scalar product $\langle u,v\rangle _{A_{0}}:=\langle u,v\rangle _{\H}+\langle A_{0}u,A_{0}v\rangle _{\H}$,  and $N$ is closed with respect to the corresponding norm, one has $\D(A_{0})=N\oplus N^{\perp}$ and so, without loss of generality,  we can suppose that $N=\ker(\tau)$, where $\tau:\D(A_{0})\to\fh$ is a bounded and surjective linear operator, $\fh$($\simeq N^{\perp}$) being an auxiliary Hilbert space, i.e. $\tau$ is a sort of a (abstract) trace map. \par In Section 2, building on \cite{P01}, we provide, by a Kre\u\i n-type resolvent formula (see Theorem \ref{teo}), the set of singular perturbations of a given self-adjoint $A_{0}$ in terms of certain families $\Lambda=\{\Lambda_{z}\}_{z\in Z_{\Lambda}}$ of bounded linear maps $\Lambda_{z}:\fb \to \fb^{*}$, where $\fb$ is a reflexive Banach space such that $\fh\hookrightarrow\fb$ is a continuous immersion and $Z_{\Lambda}$ is a not empty subset of the resolvent set of $A_{0}$. By an abstract Green-type formula, this entails the relation $\langle u,A_{\Lambda}v\rangle_{\H}=\langle A_{0}u,v\rangle_{\H}+
\langle\tau u,\varrho v\rangle_{\fh,\fh^{*}}$, where $\varrho$ is another $\fh^{*}$-valued (abstract) trace map; such a relation permits us to employ a variation (due to Schechter, see \cite{Sch} and \cite{Sch-book}) of the Cook-Kato-Kuroda method to get existence and completeness of the wave operators for the scattering couple $(A_{0},A_{\Lambda})$ in terms of conditions about the map $\tau$ and the operator family $\Lambda$ (see Theorem \ref{scattering}). \par In order to implement such conditions towards applications, in Section 3 we provide a Limiting Absorption Principle (LAP for short) holding, under certain conditions (see hypotheses H1-H4 there), for self-adjoint operators of the kind  $A_{\Lambda}$ defined in spaces of square integrable functions on arbitrary measure spaces $(M,{\mathcal B},m)$. This permits, under some further hypotheses (see hypotheses H5 and H6 in Section 3), to obtain an abstract  result about asymptotic completeness for the scattering couple $(A_{0}, A_{\Lambda})$ (see Theorem \ref{AC}).\par In Section 4, under the same hypotheses H1-H6 and using both the Birman-Kato invariance principle and Birman-Yafaev general scheme in scattering theory (see \cite{BY}, \cite{Y}, \cite{Y1}), 
we provide an explicit relation (see Corollary \ref{S-matrix}) between the Scattering Matrix $S^{\Lambda}_{\lambda}$ associated to the scattering couple $(A_{0},A_{\Lambda})$ and the limit operator $\Lambda^{+}_{\lambda}:=\lim_{\epsilon\downarrow 0}\Lambda_{\lambda+i\epsilon}$; such a limit exists in $\B(\fh,\fh^{*})$ by LAP (see Lemma \ref{bound}). \par 
Self-adjoint realizations of the Laplacian operator with boundary or interface conditions on a closed and bounded hypersurface in $\RE^{n}$ can be interpreted as singular perturbations of the free Laplacian; hence the scattering theory for these models naturally develops within the abstract scheme presented above. This point is considered in the final Section 5, where we specialize to the case in which the self-adjoint operator $A_{0}$ coincides with the free Laplacian in $L^{2}(\RE^{n})$, i.e. $A_{0}=\Delta: H^{2}(\RE^{n})\subset L^{2}(\RE^{n})\to L^{2}(\RE^{n})$, where $H^{2}(\RE^{n})$ denotes the usual Sobolev space of order two. Supposing that  to
the abstract trace map $\tau: H^{2}(\RE^{n})\to\fh$ corresponds a distribution with compact support, i.e $\ran(\tau^{*})\subseteq H_{comp}^{-2}(\RE^{n})$, and under a compactness hypothesis on $\Lambda$, we can apply our results to a wide set of singular perturbations of the free Laplacian (see Theorem \ref{LH}). Moreover, in such a setting the operator limit $\Lambda^{+}_{\lambda}$ appearing in the representation of the Scattering Matrix $S^{\Lambda}_{\lambda}$ exists in the more convenient (as regards applications) space $\B(\fb,\fb^{*})$.  In particular, we give applications to the case of scattering from Lipschitz bounded obstacles in $\RE^{n}$ both with Dirichlet (see Subsection 5.1) and Neumann (see Subsection 5.2) boundary conditions, to scattering for Schr\"odinger operators $\Delta_{\alpha}$ in $L^{2}(\RE^{n})$ with $\delta$-type potentials with unbounded strengths $\alpha$ supported on bounded $d$-sets with $0<n-d<2$ (that comprises, in the case $d=n-1$, finite unions of Lipschitz hypersurfaces which may intersect on subsets having zero ($n-1$)-dimensional measure, and, whenever $d$ is not an integer, self-similar fractals), see Subsection 5.3, and to scattering for Schr\"odinger operators $\Delta_{\theta}$ in $L^{2}(\RE^{n})$ with $\delta'$-type potentials with strength $\theta^{-1}$ supported on Lipschitz hypersurfaces (see Subsection 5.4). 
\par
Beside their interest in Quantum Mechanics, Laplace operators with boundary or interface conditions on hypersurfaces (in particular with semi-transparent boundary conditions corresponding to $\delta$ and $\delta'$ singular potentials) provide relevant models for classical scattering from obstacles or non-homogeneous acoustic media (see the recent paper \cite{Acu}). Playing a central role in direct and inverse scattering problems, the scattering amplitude (strictly related to the far-field pattern used in wave scattering, see, e.g., \cite[Chapter 6]{Kirsch}) easily derives from the $S$-matrix. Hence, our results yield to a rigorous definition and an explicit formula for this map, in the regime where the obstacles boundary or the singularity surface of the acoustic density have low regularity; this represents an important by-product and a relevant perspective of our work.\par
We conclude this introduction describing how our results extend and connect with previously known ones.  Since, by \cite{P03} (see also \cite[Theorem 2.5]{P08}), the operators $A_{\Lambda}$ have an additive representation of the kind $A_{\Lambda}=A_{0}+T_{\Lambda}$, our abstract results extend existence and completeness of scattering provided in \cite{F92} and in \cite{Bra}
for $-\Delta+\mu$, $\mu$ a signed measure (in fact Ford's paper \cite{F92} was our main inspiration in writing the present work). \par The construction developed in this work can be easily recast into the language of boundary triple theory (see \cite{DM}, \cite[Section 14]{Schm}), the maps $G_{z}$ playing the role of $\gamma$-fields and the maps $\Lambda_{z}$ being the inverses of the Weyl functions (see \cite{P04}); since we do not require any trace-class condition on resolvents differences, our results can be regarded as extensions of the abstract results  provided in \cite[Section 3]{BMN}. \par 
In Section 5 we extend to the Lipschitz case the results, there provided for smooth hypersurfaces, appearing in \cite{JST}; these already extended the results given in \cite[Section 5]{BMN}.  In more detail, the expressions for the scattering matrix we provide in \eqref{scatt-dir} relative to Dirichlet obstacles and in \eqref{scatt-neu} relative to Neumann obstacles, extend to any dimension and to Lipschitz obstacles the similar ones obtained for two-dimensional obstacles with piecewise $C^{2}$ boundary in \cite[Theorems 5.3 and 5.6]{EP1} and \cite[Theorems 4.2 and 4.3]{EP3}; similar formulae are also given, in a smooth two dimensional setting in \cite[Subsections 5.2 and 5.3]{BMN} and in a smooth $n$-dimensional setting in \cite[Subsections 6.1 and 6.2]{JST}. 
\par 
The construction of the operator $\Delta_{\alpha}$ with semi-transparent boundary condition of $\delta$-type provided in Theorem \ref{delta} extend, as regards the regularity of the boundary and/or the class of admissible strength functions, previous constructions given, for example, in \cite{Her}, \cite{F92}, \cite{BEKS}, \cite{P01}, \cite{BLL}, \cite{JDE}, \cite{ER}. Asymptotic completeness for the scattering couple $(\Delta,\Delta_{\alpha})$ provided in Theorem \ref{delta} extend results on existence and completeness given, in the case the boundary is smooth and the strength are bounded, in \cite{BLL} and \cite{JDE}. The formula for the scattering matrix provided in \eqref{scatt-delta} (respectively in \eqref{scatt-delta-lip}) extends to $d$-sets (respectively to Lipschitz hypersurfaces) the results given, in the case of a smooth hypersurface, in \cite[Subsections 6.4 and 7.4]{JST} and, in the case of a smooth $2$- or   $3$-dimensional hypersurface, in \cite[Subsection 5.4]{BMN} (see also the formula provided in \cite{F93} for Schr\"odinger operators of the kind $-\Delta+\mu$, $\mu$ a signed measure). 
\par 
The construction of the operator $\Delta_{\theta}$ with semi-transparent boundary condition of $\delta'$-type provided in Theorem \ref{delta'} extend, as regards the regularity of the boundary and/or the class of admissible strength functions, previous constructions given in \cite{BLL}, \cite{JDE}, \cite{ER}. Asymptotic completeness for the scattering couple $(\Delta,\Delta_{\theta})$ provided in Theorem \ref{delta'} extend results on existence and completeness given, whenever the boundary is smooth and $\theta$ is bounded, in \cite{BLL} and \cite{JDE}. The formula for the scattering matrix provided in \eqref{scatt-delta'} extend to Lipschitz hypersurfaces the results given, in the case of a smooth hypersurface, in \cite[Subsections 6.5 and 7.5]{JST}.      
\vskip10pt\noindent
{\bf Acknowledgments.} During the preparation of this work, the authors profited of some stays at the CNRS Institute Wolfgang Pauli of Vienna, which they gratefully acknowledge for the kind financial support.\par

\begin{subsection}{Notations.}
{\ }\par
\vskip5pt \noindent $\bullet$  $\|\cdot\|_{\X}$ denotes the norm on the complex Banach space $\X$; in case $\X$ is a Hilbert   space, $\langle\cdot,\cdot\rangle_{\X}$ denotes the (conjugate-linear w.r.t. the first argument) scalar product.
\vskip5pt\noindent $\bullet$ $\langle\cdot,\cdot\rangle_{\X^{*},\X}$ denotes the duality (assumed to be conjugate-linear w.r.t. the first argument) between the dual couple $(\X^{*},\X)$.
\vskip5pt\noindent $\bullet$ $L^{*}:\dom(L^{*})\subseteq \Y^{*}\to \X^{*}$ denotes the dual of the densely defined linear operator $L:\dom(L)\subseteq \X\to \Y$; in a Hilbert   spaces setting $L^{*}$ denotes the adjoint operator.
\vskip5pt\noindent $\bullet$ $\rho(A)$ and $\sigma(A)$ denote the resolvent set and the spectrum of the self-adjoint operator $A$; $\sigma_{p}(A)$, $\sigma_{pp}(A)$, $\sigma_{ac}(A)$,  $\sigma_{sc}(A)$, $\sigma_{ess}(A)$, $\sigma_{disc}(A)$, denote the point, pure point, absolutely continuous, singular continuous, essential and discrete spectra.
\vskip5pt\noindent $\bullet$ $\B(\X,\Y)$, $\B(\X)\equiv \B(\X,\X)$, denote the Banach space of bounded linear operator on the Banach space $\X$ to the Banach space $\Y$; ${\|}\cdot {\|}_{\X,\Y}$ denotes the corresponding norm.
\vskip5pt\noindent $\bullet$ ${\mathfrak S}_{\infty}(\X,\Y)$ denotes the space of compact operators on the Banach space $\X$ to the Banach space $\Y$.
\vskip5pt\noindent $\bullet$ $\X\hookrightarrow \Y$ means that $\X\subseteq\Y$ and for any $u\in \X$ there exists $c>0$ such that 
$\|u\|_{\Y}\le c\,\|u\|_{\X}$; we say that $\X$ is continuously embedded into $\Y$. 
\vskip5pt\noindent $\bullet$ Given the measure space $(M,{\mathcal B},m)$, $L^{2}(M,{\mathcal B},m)\equiv L^{2}(M)$ denotes the corresponding Hilbert space of measurable, square-integrable functions. 
\vskip5pt\noindent $\bullet$ $u|\Gamma$ denotes the restriction of the function $u$ to the set $\Gamma$; $L|{\mathsf V}$ denotes the restriction of the linear operator $L$ to the subspace ${\mathsf V}$. 
\vskip5pt\noindent $\bullet$ $u^{\xi}_{\lambda}$ denotes the plane wave with direction $\xi$ and wavenumber $|\lambda|^{\frac12}$, i.e. $u^{\xi}_{\lambda}(x)=e^{i\,|\lambda|^{\frac12}\xi\cdot x}$.
\vskip5pt\noindent $\bullet$  ${P}_{z}^{\-/\+}$ and ${Q}_{z}^{\-/\+}$ denote the Dirichlet-to-Neumann and Neumann-to-Dirichlet operators relative to the domain $\Omega_{\-/\+}$, where $\Omega_{\-}\equiv\Omega$ and $\Omega_{\+}:=\RE^{n}\backslash\overline\Omega$.

\end{subsection}
\end{section}
\begin{section}{Singular perturbations of self-adjoint operators.}\label{Pre1}
Given the self-adjoint operator $$A_{0}:\D(A_{0})\subseteq\H\to\H$$ in the Hilbert space $\H$ and the auxiliary Hilbert space $\fh$, let $$\tau:\D(A_0)\to\fh$$ be continuous (w.r.t. the graph norm in $\D(A_0)$) and surjective. We further assume that  ker$(\tau)$ is dense in $\H$.  For notational convenience we do not identify $\fh$ with its dual $\fh^{*}$; however we use $\fh^{**}\equiv\fh$. Typically $\fh\hookrightarrow\fh_{0}\hookrightarrow\fh^{*}$ with dense inclusions and the $\fh$-$\fh^{*}$ duality is defined in terms of the scalar product of the intermediate Hilbert space $\fh_{0}$.\par 
For any $z\in\rho(A_{0})$ we define $R^0_{z}\in\B(\H,\dom(A_0))$  by $R^{0}_{z}:=(-A_{0}+z)^{-1}$ and $G_{z}\in\B(\fh^{*},\H)$ by  
$$ G_{z}:\fh^{*}\to\H\,, \quad G_{z}:=(\tau R^{0}_{\bar z})^*\,,
$$
i.e.
\be\label{gz}
\langle G_{z}\phi,u\rangle_{\H}=\langle\phi,\tau(-A_{0}+\bar z)^{-1}u\rangle_{\fh^{*}\!,\fh}
\quad \phi\in\fh^{*}\,,\ u\in\H\,.
\ee
Since $\ker(\tau)$ is dense in $\H$, one has (see \cite[Remark 2.9]{P01},
\be\label{range}
\ran(G_{z})\cap\dom(A_0)=\{0\}\,.
\ee
However,  by the resolvent identity, 
\be\label{Gzw}
G_{z}-G_{w}=(w-z)R^0_{w}G_{z}
\ee 
and so 
\be\label{reg}
\ran(G_{z}-G_{w})\subset\dom(A_0)\,.
\ee
Notice that by \eqref{Gzw} there follows
\be\label{Gwz}
G^{*}_{z}G_{w}=G^{*}_{\bar w}G_{\bar z}\,.
\ee
Let us now suppose that there exist a reflexive Banach space $\fb\supseteq\fh$, $\fh\hookrightarrow\fb$, a set $\CO\backslash\RE\subseteq Z_{\Lambda}\subseteq\rho(A_{0})$, and a family $\Lambda$ of linear bounded maps $\Lambda_{z}\in\B(\fb ,\fb^*)$, $z\in Z_{\Lambda}$, such that (see \cite[equations (2) and (4)]{P01})
\be\label{Lambda1}
\Lambda_{z}^{*}=\Lambda_{\bar z}\,,
\ee
\be\label{Lambda2} 
\Lambda_{w}-\Lambda_{z}=(z-w)\Lambda_{w}G_{\bar w}^{*}G_{z}\Lambda_{z}\,.
\ee 
\begin{remark} In writing \eqref{Lambda2} we are implicitly using the continuous embeddings  $\fh\hookrightarrow\fb$ and  $\fb^{*}\hookrightarrow\fh^{*}$; such embeddings also give ${\|}\Lambda_{z}{\|}_{\fh,\fh^{*}}\le c\,{\|}\Lambda_{z}{\|}_{\fb,\fb^{*}}$.
\end{remark}
\begin{remark}\label{RemL} Notice that whenever $\Lambda_{z}$ has inverse $M_{z}:=\Lambda_{z}^{-1}$, then \eqref{Lambda1} and \eqref{Lambda2} are equivalent to 
\be\label{RL}
M_{z}^{*}=M_{\bar z}\,,\quad
M_{z}-M_{w}=(z-w)G_{\bar w}^{*}G_{z}\,.
\ee
\end{remark}
\begin{remark} Notice that the class of families $\Lambda$ satisfying \eqref{Lambda1} and \eqref{Lambda2} is not void: it can be parametrized by couples $(\Pi,\Theta)$, where $\Pi:\fh\to\ran(\Pi)$ is an orthogonal projection in the Hilbert space $\fh$ and $\Theta:\dom(\Theta)\subseteq\ran(\Pi)^{*}\to\ran(\Pi)$ is self-adjoint, setting (see \cite[Section 2]{P01}, \cite[Section 2]{JDE}) 
\be\label{Gamma}
\fb=\fh\,,\qquad \Lambda_{z}=\Pi^{*}(\Theta-\Pi\tau(G_{z}-(G_{z_{\circ}}+G_{\bar z_{\circ}})/2)\Pi^{*})^{-1}\Pi\,,\quad z_{\circ}\in\rho(A_{0})\,,
\ee  
$$
Z_{\Lambda}=Z_{\Pi,\Theta}:=\{z\in\rho(A_{0}):\text{$\Theta-\Pi\tau(G_{z}-(G_{z_{\circ}}+G_{\bar z_{\circ}})/2)\Pi^{*}$ has a bounded inverse}\}\,.
$$
The set $Z_{\Pi,\Theta}$ always contains $\CO\backslash\RE$ (see the proof of \cite[Theorem 2.1]{P08}; see also \cite[Proposition 2.1]{P01}) and so it is not void. In concrete situations it could happen that it is  better to 
work with different representations and/or to choose a space $\fb$ strictly larger than $\fh$; then \eqref{Lambda1} and \eqref{Lambda2} have to be checked case by case.
\end{remark}
\begin{theorem}\label{teo}
Let $\Lambda$ satisfy \eqref{Lambda1} and \eqref{Lambda2}. Then the family of bounded linear maps $R^{\Lambda}_{z}\in \B(\H)$, $z\in  Z_{\Lambda}$, defined by 
\be\label{resolvent}
R_{z}^{\Lambda}:=R^0_{z}+G_{z}\Lambda_{z}G^{*}_{\bar z}\,.
\ee
is the resolvent of a self-adjoint operator $A_{\Lambda}$ which is a self-adjoint extension of the closed symmetric operator $S:=A_{0}|\ker(\tau)$. 
\end{theorem}
\begin{proof} We proceed as in the proof of \cite[Theorem 2.1]{P01}. By \eqref{Lambda2}, $z\mapsto R^{\Lambda}_{z}$ is a pseudo-resolvent, i.e. it satisfies the resolvent identity (see \cite[page 113]{P01}). Since, by \eqref{resolvent}, $u\in\ker (R^{\Lambda}_{z})$ gives $R^0_{z}u\in\ran(G_{z})$, one gets $u=0$ by \eqref{range} and so $R^{\Lambda}_{z}$ is injective. Moreover, by \eqref{Lambda1} one gets $(R^{\Lambda}_{z})^{*}=R^{\Lambda}_{\bar z}$. Thus, by \cite[Theorems 4.10 and 4.19]{Stone}, $R^{\Lambda}_{z}$ is the resolvent of a self-adjoint operator $A_{\Lambda}$. Let us now fix $z\in Z_{\Lambda}$. By \eqref{range} and by 
\be\label{domLambda}
\dom(A_{\Lambda})=\ran(R^{\Lambda}_{z})=\{u=u_z+G_{z}\Lambda_{z}\tau u_{z},\ u_{z}\in\dom(A_0) \}\,,
\ee
one gets 
\be\label{dom}
\dom(A_0)\cap\dom(A_{\Lambda})=\{u\in\dom(A_0):\Lambda_{z}\tau u=0\}\,.
\ee
Thus, by  
\be\label{op}
(-A_{\Lambda}+z)u=(R^{\Lambda}_{z})^{-1}u=(-A_{0}+z)u_{z}\,,
\ee 
one gets 
$A_{\Lambda}|\ker(\tau)=A_{0}|\ker(\tau)=S$.
\end{proof}
\begin{remark}\label{rem} By the above proof there follows that Theorem \ref{teo} holds true without requiring that $Z_{\Lambda}$ contains the whole $\CO\backslash\RE$: it suffices to suppose that $Z_{\Lambda}\subseteq \rho(A_{0})$ is a not empty set which is symmetric with respect to the real axis. However, the former hypothesis is used in our successive treatments of Scattering Theory and Limiting Absorption Principle.
\end{remark}
\begin{remark} By \eqref{dom} and \eqref{op}, one gets 
$$
\ker(\tau)\subseteq\{u\in\D(A_0)\cap\D(A_{\Lambda}):Au=A_{\Lambda}u\}\,.
$$
Since  $\ker(\tau)$ is dense by our hypothesis, the set $\{u\in\D(A_0)\cap\D(A_{\Lambda}):A_{0}u=A_{\Lambda}u\}$ is dense as well and so $A_{\Lambda}$ is a singular perturbation of $A_{0}$ (see \cite{P04}). 
\end{remark}
\begin{remark}
By \cite[Corollary 3.2]{P08} (see also \cite[Theorem 2.1]{JDE}), the representation \eqref{Gamma} shows that any self-adjoint extension of $S$ is of the kind provided in Theorem \ref{teo}. 
\end{remark}
Now, in order to simplify the exposition and since such an hypothesis holds true in the  
applications further considered, we suppose that $A_{0}$ has a spectral gap, i.e. 
$$
\rho(A_{0})\cap \RE\not=\emptyset\,.
$$
Then, we pick $\lambda_{\circ}\in \rho(A_{0})\cap\RE$ and set
\be
G_{\circ} :=G_{\lambda_{\circ}}\,.
\ee
Let $S$ be the symmetric operator defined by $S:=A_{0}|\ker{\tau}$ as in Theorem \ref{teo}. By \cite[Theorem 3.1]{P04} and \cite[Lemma 2.3 and Remark 2.4]{JDE}, one has (compare with \eqref{domLambda} and \eqref{op})
\be\label{dom-adj}
\dom(S^{*})=\{u=u_{\circ}+G_{\circ} \phi\,,\ u_{\circ}\in\D(A_0)\,,\ \phi\in\fh^{*}\}\,,
\ee
\be\label{adj}
(-S^{*}+\lambda_{\circ})u=(-A_{0}+\lambda_{\circ})u_{\circ}\,,
\ee
(one can check that the definition of $S^{*}$ is $\lambda_{\circ}$-independent) 
and, defining the bounded linear map 
\be\label{varrho}
\varrho:\D(S^{*})\to\fh^{*}\,,\quad\varrho  u\equiv \varrho(u_{\circ}+G_{\circ} \phi):=\phi\,,
\ee
the following Green's type identity holds (see \cite[Theorem 3.1]{P04}, \cite[Remark 2.4]{JDE}):
\be\label{green1}
\langle u,S^{*}v\rangle_{\H}-\langle S^{*}u,v\rangle_{\H}=
\langle\tau u_{\circ},\varrho v\rangle_{\fh,\fh^{*}}-
\langle\varrho u,\tau v_{\circ}\rangle_{\fh^{*}\!,\fh}\,,\quad u,v\in\dom(S^{*})\,.
\ee
Thus, in particular, since $A_{0}\subset S^{*}$, $\dom(A_{0})=\ker(\varrho)$ and $A_{\Lambda}\subset S^{*}$,
\be\label{green2}
\langle u,A_{\Lambda}v\rangle_{\H}=\langle A_{0}u,v\rangle_{\H}+
\langle\tau u,\varrho v\rangle_{\fh,\fh^{*}}\,,\quad u\in\dom(A_0)\,,\ v\in\dom(A_{\Lambda})\,.
\ee
The identity \eqref{green2} is our starting point for the following abstract result about scattering for the couple $(A_{0},A_{\Lambda})$:
\begin{theorem}\label{scattering} Let $A_{\Lambda}$ be defined according to Theorem \ref{teo}. Suppose that there exists an open subset $\Sigma\subseteq\RE$ of full measure such that for any open and bounded $I$, $\overline{I}\subset \Sigma$, 
\be\label{in1}
\sup_{(\lambda,\epsilon)\in I\times (0,1)}\,\epsilon^{\frac12}\,{\|}G_{\lambda\pm i\epsilon}{\|}_{\fh^{*},\H}<+\infty\,,
\ee
and
\be\label{in2}
\sup_{(\lambda,\epsilon)\in I\times (0,1)}\,{\|}\Lambda_{\lambda\pm i\epsilon}{\|}_{\fh,\fh^{*}}<+\infty\,.
\ee
Then the strong limits 
$$
W_{\pm}(A_ {\Lambda},A_{0})
:=\text{s-}\lim_{t\to\pm\infty}e^{-itA_{\Lambda}}e^{itA_{0}}P^{0}_{ac}
\,,\qquad
W_{\pm}(A_{0},A_{\Lambda})
:=\text{s-}\lim_{t\to\pm\infty}e^{-itA_{0}}e^{itA_{\Lambda}}P^{\Lambda}_{ac}\,,
$$ 
exist everywhere in $\H$ and are complete, i.e. $$\ran(W_{\pm}(A_ {\Lambda},A_{0}))=\H^{\Lambda}_{ac}\,,\qquad\ran(W_{\pm}(A_{0},A_{\Lambda}))=\H^{0}_{ac}\,,$$
$$W_{\pm}(A_ {\Lambda},A_{0})^{*}=W_{\pm}(A_{0},A_ {\Lambda})\,,$$ where $P^{0}_{ac}$ and $P^{\Lambda}_{ac}$ are the orthogonal projectors onto $\H^{0}_{ac}$ and $\H^{\Lambda}_{ac}$, the absolutely continuous subspaces relative to $A_{0}$ and $A_{\Lambda}$ respectively. 
\end{theorem}
\begin{proof} At first let us show that \eqref{in1} and \eqref{in2} imply 
\be\label{in2*}
\sup_{(\lambda,\epsilon)\in I\times (0,1)}\left (\epsilon^{\frac12}\,{\|}\tau R^{0}_{\lambda\pm i\epsilon}{\|}_{\H,\fh}+\epsilon^{\frac12}\,{\|}\varrho R^{\Lambda}_{\lambda\pm i\epsilon}{\|}_{\H,\fh^{*}}\right)<+\infty\,.
\ee
By the definition of $G_{z}$ one has 
$$
{\|}\tau R^{0}_{\bar z}{\|}_{\H,\fh}={\|}G^{*}_{\bar z}{\|}_{\H,\fh}={\|}G_{z}{\|}_{\fh^{*},\H}\,.
$$
By \eqref{resolvent} and \eqref{reg}, one has
\begin{align*}
&\varrho R^{\Lambda}_{z}=\varrho (R^0_{z}+G_{z}\Lambda_{z}G_{\bar z}^{*})=
\varrho (R^0_{z}u+(G_{z}-G_{\circ})\Lambda_{z}G_{\bar z}^{*}u+G_{\circ}\Lambda_{z}G_{\bar z}^{*})=\Lambda_{z}G_{\bar z}^{*}
\end{align*}
and so 
\begin{align}\label{rhoR}
&{\|}\varrho R^{\Lambda}_{z} {\|}_{\H,\fh^{*}}={\|}\Lambda_{z}G_{\bar z}^{*}{\|}_{\H,\fh^{*}}\le
{\|}\Lambda_{z}{\|}_{\fh,\fh^{*}}{\|}G_{\bar z}^{*}{\|}_{\H,\fh}=
{\|}\Lambda_{z}{\|}_{\fh,\fh^{*}}{\|}G_{z}{\|}_{\fh^{*},\H}
\,.
\end{align}
Now we follow the same reasonings as in the proof of \cite[Theorem 9.4.2]{Sch-book} (see also \cite[Section 3]{Sch}). At first let us notice that in our setting equation \eqref{green2} agree with \cite[equation (9.4.1)]{Sch-book} whenever the operators  there denoted  by $A$ and $B$ correspond to $\tau$ and $\iota^{-1}\rho$ respectively, where $\iota:\fh\to\fh^{*}$ is the duality mapping given by the canonical isomorphism from $\fh$ onto $\fh^{*}$; therefore \eqref{in2*} corresponds to \cite[estimate (9.4.8)]{Sch-book}. Thus (compare with the first lines of the proof of \cite[Theorem 9.4.2]{Sch-book}), \cite[Lemma 9.3.3, Corollary 9.3.1 and Lemma 9.3.2]{Sch-book} give, for any $u^{0}_{c}\in \H^{0}_{c}$ and $u^{\Lambda}_{c}\in \H^{\Lambda}_{c}$,
\be\label{in1*}
\int_{-\infty}^{+\infty}\left(\|\tau e^{-itA_{0}}E_{0}(I)u^{0}_{c}\|^{2}_{\fh}+\|\varrho e^{-itA_{\Lambda}}E_{\Lambda}(I)u^{\Lambda}_{c}\|^{2}_{\fh^{*}}\right)dt<+\infty\,,
\ee
where $E_{0}$, $\H^{0}_{c}$ and $E_{\Lambda}$, $\H^{\Lambda}_{c}$ denote the spectral measures and the continuous subspaces relative to $A_{0}$ and $A_{\Lambda}$ respectively. According to \cite[Theorem 9.4.1]{Sch-book}, \eqref{in2*} and \eqref{in1*} give 
$$
E_{0}(\Sigma)u^{0}_{c}\in M_{\pm}(A_{\Lambda}, A_{0}):=\{u\in\H: \text{ $\lim_{t\to\pm\infty}e^{-itA_{\Lambda}}e^{itA_{0}}u$ exists} \}
$$
and
$$
E_{\Lambda}(\Sigma)u^{\Lambda}_{c}\in M_{\pm}(A_{0},A_{\Lambda}):=\{u\in\H: \text{ $\lim_{t\to\pm\infty}e^{-itA_{0}}e^{itA_{\Lambda}}u$ exists} \}\,.
$$
Thus, by $\H^{0}_{ac}\subseteq \H^{0}_{c}$, $\H^{\Lambda}_{ac}\subseteq \H^{\Lambda}_{c}$, and, since $\Sigma^{c}$ has Lebesgue measure zero, by $E_{0}(\Sigma)P^{0}_{ac}=P^{0}_{ac}$,  $E_{\Lambda}(\Sigma)P^{\Lambda}_{ac}=P^{\Lambda}_{ac}$, both the wave operators 
$W_{\pm}(A_ {\Lambda},A_{0})$ and $W_{\pm}(A_{0},A_ {\Lambda})$ exist; this  also gives    completeness (see e.g. \cite[Proposition 3, Section XI.3]{RS-III}). 
\end{proof}
\end{section}
\begin{section}{The Limiting Absorption Principle and Asymptotic completeness.}
Now we suppose that $\H=L^{2}(M,{\mathcal B},m)\equiv L^{2}(M)$. 
Given a measurable $\varphi:M\to (0,+\infty)$, we define the weighted $L^{2}$-space 
\be\label{Lphi}
L_{\varphi}^{2}(M,{\mathcal B},m)\equiv L_{\varphi}^{2}(M):=\{\text{$u:M\to\CO$ measurable} : \varphi u\in L^{2}(M)\}\,.
\ee
From now on $\langle\cdot,\cdot\rangle $ and  $\|\cdot\|$ denote the scalar product and the corresponding norm on  $L^{2}(M)$; $\langle\cdot,\cdot\rangle_{\varphi} $ and  $\|\cdot\|_{\varphi}$ denote the scalar product and the corresponding norm on  $L_{\varphi}^{2}(M)$. In the following we suppose $\varphi\ge 1$ $m$-a.e.; therefore 
$$
L^{2}_{\varphi}(M)\hookrightarrow L^{2}(M)\hookrightarrow L_{\varphi^{-1}}^{2}(M)\simeq L_{\varphi}^{2}(M)^{*}\,.
$$
Then we introduce the following hypotheses: \vskip5pt\par\noindent
(H1) both $A_{0}$ and $A_{\Lambda}$ are bounded from above and there exists $c_{1}>0$ such that the maps 
$z\mapsto R^0_{z}$ and $z\mapsto R^{\Lambda}_{z}$ are continuous on $\{z\in\CO:\text{Re$(z)>c_{1}$}\}$ to $\B(L_{\varphi}^{2}(M))$;
\vskip5pt\par\noindent
(H2) $A_{0}$ satisfies a Limiting Absorption Principle (LAP for short), i.e. there exists an open set $\Sigma_{0}\subseteq \RE$ of full measure such that for all $\lambda\in \Sigma_{0}$ the limits
\be\label{LAP1}
R^{0,\pm}_{\lambda}:=\lim_{\epsilon\downarrow 0}R^{0}_{\lambda\pm i\epsilon}
\ee
exist in $\B(L_{\varphi}^{2}(M),L_{\varphi^{-1}}^{2}(M))$ and the maps $z\mapsto R^{0,\pm}_{z}$, where $R^{0,\pm}_{z}\equiv R^{0}_{z}$ whenever $z\in\rho(A_{0})$, are continuous on $\Sigma_{0}\cup\CO_{\pm}$ to $\B(L_{\varphi}^{2}(M),L_{\varphi^{-1}}^{2}(M))$;
\vskip5pt\par\noindent
(H3) for any compact set $K\subset \Sigma_{0}$ there exists $c_{K}>0$ such that for any $\lambda\in K$ and for any $u\in L_{\varphi^{2}}^{2}(M)\cap\ker(R^{0,+}_{\lambda}-R^{0,-}_{\lambda})$ one has
\be\label{(H3)}
\|R^{0,\pm}_{\lambda}u\|\le c_{K}\|u\|_{\varphi^{2}}\,;
\ee 
\vskip5pt\par\noindent
(H4) there exist $c_{2}>0$, $\gamma>0$ and $k\in\NA$ such that  for all $\lambda>c_{2}$
\be\label{(H4)}
(R^{\Lambda}_{\lambda})^{k}-(R^{0}_{\lambda})^{k}\in {\mathfrak S}_{\infty}(L^{2}(M),L_{\varphi^{2+\gamma}}^{2}(M))\,.
\ee
\end{section}
Then $A_{\Lambda}$ satisfies a Limiting Absorption Principle as well:
\begin{theorem}\label{LAP} Suppose hypotheses (H1)-(H4) hold. Then $\Sigma_{\Lambda}:=\Sigma_{0}\cap \sigma_{p}(A_{\Lambda})$ is a (possibly empty) discrete set and for all $\lambda\in \Sigma_{0}\backslash \Sigma_{\Lambda}$ the limits
\be\label{LAP2}
R^{\Lambda,\pm}_{\lambda}:=\lim_{\epsilon\downarrow 0}R^{\Lambda}_{\lambda\pm i\epsilon}
\ee
exist in $\B(L_{\varphi}^{2}(M),L_{\varphi^{-1}}^{2}(M))$; the maps $z\mapsto R^{\Lambda,\pm}_{z}$, where $R^{\Lambda,\pm}_{z}\equiv R^{\Lambda}_{z}$ whenever $z\in\rho(A_{\Lambda})$, are continuous on $(\Sigma_{0}\backslash \Sigma_{\Lambda})\cup\CO_{\pm}$ to $\B(L_{\varphi}^{2}(M),L_{\varphi^{-1}}^{2}(M))$
\end{theorem}
\begin{proof} Hypotheses (H1)-(H4) permit us to use the abstract results contained in \cite{Ren},  following the same argumentations provided in the proof of \cite[Theorem 4.2]{JST}:  hypotheses (T1) and (E1) in \cite[page 175]{Ren} correspond to our (H2), (H3) and (H4); then, by \cite[Proposition 4.2]{Ren}, the latter imply hypotheses (LAP), (E) in \cite[page 166]{Ren} and hypothesis (T) in \cite[page 168]{Ren}. In our setting (LAP), (E) and (T) correspond  respectively to (H2), 
$$
(R^{\Lambda}_{\lambda})^{k}-(R^{0}_{\lambda})^{k}\in {\mathfrak S}_{\infty}(L_{\varphi^{-1}}^{2}(M),L_{\varphi}^{2}(M))\,,
$$ 
and a technical variant of (H3). According to \cite[Theorem 3.5]{Ren}, the three hypotheses (LAP), (E) and (T), together with (H1) (i.e. hypothesis (OP) in \cite[page 165]{Ren}), give the thesis.  
\end{proof}
\begin{remark} In order to get Theorem \ref{LAP} one does not need to require $\Sigma_{0}$ to be a set of full measure. However that hypothesis is needed for next Theorem \ref{AC}.
\end{remark}
Since, as is well known, LAP implies absence of singular continuous spectrum (see e.g. \cite[Theorem 6.1]{Agm}, \cite[Corollary 4.7]{JST}), by (H2) and Theorem \ref{LAP}, one gets
\begin{corollary}\label{sc} Suppose hypotheses (H1)-(H4) hold. Then
\be
\sigma_{sc}(A_{0})=\sigma_{sc}(A_{\Lambda})=\emptyset\,.
\ee
Equivalently
\be
(L^{2}(M)^{0}_{pp})^{\perp}=L^{2}(M)^{0}_{ac}\,,\quad 
(L^{2}(M)^{\Lambda}_{pp})^{\perp}=L^{2}(M)^{\Lambda}_{ac}\,,
\ee
where  $L^{2}(M)^{0}_{pp}$, $(L^{2}(M)^{0}_{ac}$, $L^{2}(M)^{\Lambda}_{pp}$, $(L^{2}(M)^{\Lambda}_{ac}$ denotes the pure point and absolutely continuous subspaces of $L^{2}(M)$ with respect to $A_{0}$ and $A_{\Lambda}$. 
\end{corollary}
Let us now introduce the further hypothesis H5, which, 
for future convenience, we split in two separate assumptions:\vskip8pt\noindent
(H5.1) \text{for any $z\in\rho(A_{0})$, $G_{z}^{*}=\tau R_{\bar z}^{0}:L^{2}_{\varphi}(M)\to\fh$ is surjective;}\vskip8pt\noindent
(H5.2) for any $\lambda\in\Sigma_{0}$, the limits 
\be\label{limG*}
(G_{\lambda}^{\pm})^{*}:=\lim_{\epsilon\downarrow 0}\tau R^{0}_{\lambda\mp i\epsilon}
\ee
exist in $\B(L^{2}_{\varphi}(M), \fh)$ and are surjective; moreover the maps $z\mapsto (G_{z}^{\pm})^{*}$, where $(G_{z}^{\pm})^{*}\equiv G_{z}^{*}$ whenever $z\in\rho(A_{0})$, are continuous on $\Sigma_{0}\cup\CO_{\mp}$ to $\B(L^{2}_{\varphi}(M), \fh)$.
\vskip8pt\noindent
\begin{remark}\label{3.4} By hypothesis (H5) and by duality, for any $\lambda\in\Sigma_{0}$, the limits 
\be\label{limG}
G_{\lambda}^{\pm}:=\lim_{\epsilon\downarrow 0}(\tau R^{0}_{\lambda\mp i\epsilon})^{*}
\ee
exist in $\B(\fh^{*},L^{2}_{\varphi^{-1}}(M))$ and are injective; moreover the maps $z\mapsto G_{z}^{\pm}$, where $G_{z}^{\pm}\equiv G_{z}$ whenever $z\in\rho(A_{0})$, are continuous on $\Sigma_{0}\cup\CO_{\pm}$ to $\B(\fh^{*},L^{2}_{\varphi^{-1}}(M))$.

\end{remark}
\begin{remark}\label{RMM} Here we recall the definition of {\it reduced minimum modulus} $\gamma(T)$ of a linear operator $T\in\B(X,Y)$, $T\not=0$: 
\be\label{min-mod}
\gamma(T):=\inf\{\|Tu\|_{Y}:\text{dist}_{X}(u,\ker(T))=1\}\,.
\ee  
By \cite[Theorem 5.2, page 231]{K}, $T$ has closed range if and only if $\gamma(T)>0$. Moreover, see \cite[Proposition 1.1]{Ap}, 
\be\label{cont}
\ker(T_{1})=\ker(T_{2})\quad\Longrightarrow\quad
|\gamma(T_{1})-\gamma(T_{2})|\le {\|}T_{1}-T_{2}{\|}_{X,Y}\,.
\ee 
\end{remark}

Then, by \eqref{LAP1}, \eqref{LAP2} and by the resolvent formula \eqref{resolvent}, one gets the following
\begin{lemma}\label{bound}
Suppose hypotheses (H1)-(H5) hold. Then, for any open and bounded $I$, $\overline{I}\subset \Sigma_{0}\backslash\Sigma_{\Lambda}$, one has 
\be\label{supLambda}
\sup_{(\lambda,\epsilon)\in I\times (0,1)}{\|}\Lambda_{\lambda\pm i\epsilon}{\|}_{\fh,\fh^{*}}<+\infty\,.
\ee
Moreover, for any $\lambda\in \Sigma_{0}\backslash\Sigma_{\Lambda}$, the limits 
\be\label{limLambda}
\Lambda^{\pm}_{\lambda}:=\lim_{\epsilon\downarrow 0}\Lambda_{\lambda\pm i\epsilon}\,.
\ee
exist in $\B(\fh,\fh^{*})$ and
\be\label{limRes}
R^{\Lambda,\pm}_{\lambda}=R^{0,\pm}_{\lambda}+G^{\pm}_{\lambda}\Lambda^{\pm}_{\lambda}(G^{\mp}_{\lambda})^{*}\,.
\ee
\end{lemma}
\begin{proof} 
Since $(G_{z}^{\pm})^{*}$ are surjective by hypotheses (H5), $G_{z}^{\pm}$  are injective and have closed range by the closed range theorem. Hence, by Remark \ref{RMM}, $\gamma(G^{\pm}_{z})>0$, and, since $\ker(G^{\pm}_{z})=\{0\}$ for any $z\in \Sigma_{0}\cup\CO_{\pm}$,  the maps $z\mapsto \gamma(G^{\pm}_{z})$ are continuous on $\Sigma_{0}\cup\CO_{\pm}$ to $(0,+\infty)$ by \eqref{cont}. Hence, by \eqref{min-mod}, for any open and bounded $I$, $\overline{I}\subset \Sigma_{0}$, for any $(\lambda,\epsilon)\in I\times(0,1)$  and for any $\phi\in \fh^{*}$, there exist $\gamma^{\pm}_{I}>0$ such that 
\begin{align*}
&\|G_{\lambda\pm i\epsilon}\phi\|_{\varphi^{-1}}
\ge \gamma_{I}^{\pm}\|\phi\|_{\fh^{*}}\,.
\end{align*}
Therefore, by \eqref{resolvent} and \eqref{Lambda1},
\begin{align*}
{\|}R_{\lambda\mp i\epsilon}^{\Lambda}-R^{0}_{\lambda\mp i\epsilon}{\|}_{L^{2}_{\varphi},L^{2}_{\varphi^{-1}}}=&
{\|}G_{\lambda\mp i\epsilon}\Lambda_{\lambda\mp i\epsilon}G^{*}_{\lambda\pm i\epsilon}{\|}_{L^{2}_{\varphi},L^{2}_{\varphi^{-1}}}\ge 
\gamma_{I}^{\mp}{\|}\Lambda_{\lambda\mp i\epsilon}G^{*}_{\lambda\pm i\epsilon}{\|}_{L^{2}_{\varphi},\fh^{*}}\\
=&
\gamma^{\mp}_{I}{\|}G_{\lambda\pm i\epsilon}\Lambda_{\lambda\pm i\epsilon}{\|}_{\fh,L^{2}_{\varphi^{-1}}}\ge 
\gamma_{I}^{\mp}\gamma_{I}^{\pm}{\|}\Lambda_{\lambda\pm i\epsilon}{\|}_{\fh,\fh^{*}}\,.
\end{align*} 
Then, by hypothesis (H2) and Theorem \ref{LAP}, one gets \eqref{supLambda}. \par
By hypothesis (H2), Theorem \ref{LAP}, hypothesis (H5.2), Remark \ref{3.4} and \eqref{supLambda}, one obtains the existence and equality in $\B(L^{2}_{\varphi}(\RE^{n}),L^{2}_{\varphi^{-1}}(\RE^{n}))$ of the limits
\be\label{LGR}
\lim_{\epsilon\downarrow 0}G_{\lambda\pm i\epsilon}\Lambda_{\lambda\pm i\epsilon}
G^{*}_{\lambda\mp i\epsilon}=
\lim_{\epsilon\downarrow 0}G^{\pm}_{\lambda}\Lambda_{\lambda\pm i\epsilon}
(G^{\mp}_{\lambda})^{*}=R_{\lambda}^{\Lambda,\pm}-R^{0,\pm}_{\lambda}\,.
\ee
Then, proceeding in a similar way as above, one has 
\begin{align*}
&{\|}G^{\pm}_{\lambda}(\Lambda_{\lambda\pm i\epsilon_{1}}-\Lambda_{\lambda\pm i\epsilon_{2}})(G^{\mp}_{\lambda})^{*}{\|}_{L^{2}_{\varphi},L^{2}_{\varphi^{-1}}}\ge 
\gamma(G^{\pm}_{\lambda})\,{\|}(\Lambda_{\lambda\pm i\epsilon_{1}}-\Lambda_{\lambda\pm i\epsilon_{2}})(G^{\mp}_{\lambda})^{*}{\|}_{L^{2}_{\varphi},\fh^{*}}\\
=&\gamma(G^{\pm}_{\lambda})\,{\|}G^{\mp}_{\lambda}(\Lambda_{\lambda\mp i\epsilon_{1}}-\Lambda_{\lambda\mp i\epsilon_{2}}){\|}_{\fh,L^{2}_{\varphi^{-1}}}
\ge \gamma(G^{\pm}_{\lambda})\gamma(G^{\mp}_{\lambda})\,{\|}\Lambda_{\lambda\mp i\epsilon_{1}}-\Lambda_{\lambda\mp i\epsilon_{2}}{\|}_{\fh,\fh^{*}}\,.
\end{align*}
This and \eqref{LGR} give the existence of the limits \eqref{limLambda} and then the limit resolvent formulae \eqref{limRes}.
\end{proof}
Our last hypothesis is the following: \vskip8pt\noindent
(H6) for any $z\in\rho(A_{0})$, $G_{z}\in\B(\fh^{*},L^{2}_{\varphi}(M))$.
\vskip8pt\noindent
\begin{remark} By duality, hypothesis (H6) is equivalent to requiring that $\tau R^{0}_{z}$  has a bounded extension on $L^{2}_{\varphi^{-1}}(M)$ to $\fh$ for any $z\in\rho(A_{0})$. 
\end{remark}
\begin{remark}\label{R(H6)} By \eqref{resolvent} and (H6), if $R^{0}_{z}\in \B(L^{2}_{\varphi}(M))$ for any $z\in\CO$ such that $\text{Re}(z)>c_{1}$, then the same is true for $R^{\Lambda}_{z}$. Thus the maps $z\mapsto R^{0}_{z}$ and $z\mapsto R^{\Lambda}_{z}$ satisfy hypothesis (H1) (they are continuous since pseudo-resolvents in $\B(L^{2}_{\varphi}(M))$).
\end{remark}
Then the previous results lead to 
\begin{theorem}\label{AC} Suppose that the couple $(A_{0},A_{\Lambda})$ satisfies hypotheses (H1)-(H6). Then asymptotic completeness holds, i.e. the strong limits 
$$
W_{\pm}(A_ {\Lambda},A_{0})
:=\text{s-}\lim_{t\to\pm\infty}e^{-itA_{\Lambda}}e^{itA_{0}}P^{0}_{ac}
\,,\qquad
W_{\pm}(A_{0},A_{\Lambda})
:=\text{s-}\lim_{t\to\pm\infty}e^{-itA_{0}}e^{itA_{\Lambda}}P^{\Lambda}_{ac}\,,
$$ 
exist everywhere in $L^{2}(M)$, $$\ran(W_{\pm}(A_ {\Lambda},A_{0}))=(L^{2}(M)^{\Lambda}_{pp})^{\perp}\,,\qquad\ran(W_{\pm}(A_{0},A_{\Lambda}))=(L^{2}(M)^{0}_{pp})^{\perp}\,,
$$
$$ W_{\pm}(A_ {\Lambda},A_{0})^{*}=W_{\pm}(A_{0},A_ {\Lambda})\,,$$ 
where $P^{0}_{ac}$ and $P^{\Lambda}_{ac}$ are the orthogonal projectors onto the absolutely continuous subspaces $L^{2}(M)^{0}_{ac}$ and $L^{2}(M)^{\Lambda}_{ac}$. 
\end{theorem}
\begin{proof} By Theorem \ref{scattering}, to get completeness we need to show that \eqref{in1} and \eqref{in2} hold true. Then, asymptotic completeness is consequence of Corollary \ref{sc}. The bound \eqref{in2} is given in Lemma \ref{bound} and so we just need to prove  the bound \eqref{in1}. Let $\mu\in \rho(A_{0})\cap\RE$; by \eqref{Gzw}, \eqref{Gwz} and \eqref{reg}, one has
\begin{align*}
&|\text{Im}(z)|\,\|G_{\bar z}\phi\|^{2}
=\left|\text{Im}(z)\langle G_{\bar z}^{*}G_{\bar z}\phi,\phi \rangle_{\fh,\fh^{*}}\right|=
\frac{1}2\,\left|\langle\tau(G_{\bar z}-G_{z})\phi,\phi\rangle_{\fh,\fh^{*}}\right|\\
=&\frac{1}2\,\left|\big(\langle\tau(G_{\bar z}-G_{\mu})\phi,\phi\rangle_{\fh,\fh^{*}}-
\langle\tau(G_{z}-G_{\mu})\phi,\phi\rangle_{\fh^{*}\!,\fh}\big)\right|\\
=&\frac{1}2\,\left|\big((\mu- z)\langle G_{z}^{*}G_{\mu}\phi,\phi\rangle_{\fh,\fh^{*}}-(\mu-\bar z)\langle G_{\bar z}^{*}G_{\mu}\phi,\phi\rangle_{\fh,\fh^{*}}\big)\right|\\
\le &\frac{1}2\,|\mu- z|\,\big({\|} G_{z}^{*}{\|}_{L^{2}_{\varphi},\fh}+
{\|} G_{\bar z}^{*}{\|}_{L^{2}_{\varphi},\fh}\big){\|} G_{\mu}{\|}_{\fh^{*}\!\!,L^{2}_{\varphi}}\|\phi\|^{2}_{\fh^{*}}\,.
\end{align*}
Therefore
\begin{align}\label{eps}
\epsilon\,{\|}G_{\lambda\pm i\epsilon}{\|}_{\fh^{*}\!\!,L^{2}}^{2}
\le \frac12\left(|\mu-\lambda|+\epsilon\right){\|} G_{\mu}{\|}_{\fh^{*}\!\!,L^{2}_{\varphi}}\,\big({\|} \tau R^{0}_{\lambda+i\epsilon}{\|}_{L^{2}_{\varphi},\fh}+
{\|} \tau R^{0}_{\lambda-i\epsilon}{\|}_{L^{2}_{\varphi},\fh}\big)
\end{align}
and the bound \eqref{in1} is consequence of hypotheses (H5.2) and (H6).
\end{proof}
\begin{section}{The Scattering Matrix.} 
According to Theorem \ref{AC}, under hypotheses (H1)-(H6), the scattering operator 
$$
S_{\Lambda}:=W_{+}(A_{\Lambda},A_{0})^{*}W_{-}(A_{\Lambda},A_{0})
$$
is a well defined unitary map. Given a direct integral representation of $L^{2}(M)^{0}_{ac}$ with respect to the spectral measure of the absolutely continuous component of $A_{0}$  (see e.g. \cite[Section 4.5.1]{BW} ), i.e. a unitary map 
\be\label{F0}
F_{0}:L^{2}(M)^{0}_{ac}\to\int^{\oplus}_{\sigma_{ac}(A_{0})}(L^{2}(M)^{0}_{ac})_{\lambda}\,d\eta(\lambda)
\ee
which diagonalizes the absolutely continuous component of $A_{0}$, 
we define the scattering matrix  $$S^{\Lambda}_{\lambda}:(L^{2}(M)^{0}_{ac})_{\lambda}\to (L^{2}(M)^{0}_{ac})_{\lambda}$$ by the relation (see e.g. \cite[Section 9.6.2]{BW})
$$
F_{0}S_{\Lambda}F_{0}^{*}u_{\lambda}=S^{\Lambda}_{\lambda}u_{\lambda}
\,.
$$
Now, following the same scheme as in \cite[Remark 5.7]{JST}, which uses the Birman-Kato invariance principle and the Birman-Yafaev general scheme in stationary scattering theory (see e.g. \cite{BY}, \cite{Y}, \cite{Y1}), we provide an explicit relation between  $S^{\Lambda}_{\lambda}$ and $\Lambda^{+}_{\lambda}$. \par 
Given $\mu\in \rho(A_{0})\cap\rho(A_{\Lambda})$, we consider the scattering couple $(R^{\Lambda}_{\mu}, R^{0}_{\mu})$ and the strong limits 
$$
W_{\pm}(R^{\Lambda}_{\mu},R^{0}_{\mu})
:=\text{s-}\lim_{t\to\pm\infty}e^{-itR^{\Lambda}_{\mu}}e^{itR^{0}_{\mu}}P^{\mu}_{ac}\,,
$$ 
where $P^{\mu}_{ac}$ is the orthogonal projector onto the absolutely continuous subspace of $R^{0}_{\mu}$; we prove below that such limits exist everywhere in $L^{2}(M)$. Let $S^{\mu}_{\Lambda}$ the corresponding scattering operator  
$$
S_{\Lambda}^{\mu}:=W_{+}(R^{\Lambda}_{\mu},R^{0}_{\mu})^{*}W_{-}(R^{\Lambda}_{\mu},R^{0}_{\mu})\,.
$$
Using the unitary operator $F_{0}^{\mu}$ which diagonalizes the absolutely continuous component of $R^{0}_{\mu}$, i.e. $(F^{\mu}_{0}u)_{\lambda}:=\frac1\lambda(F_{0}u)_{\mu-\frac1\lambda}$, $\lambda\not=0$ such that $\mu-\frac1\lambda\in \sigma_{ac}(A_{0})$, one defines the scattering  matrix  $$S^{\Lambda,\mu}_{\lambda}:(L^{2}(M)^{0}_{ac})_{\mu-\frac1\lambda}\to (L^{2}(M)^{0}_{ac})_{\mu-\frac1\lambda}$$  corresponding to the scattering operator $S^{\mu}_{\Lambda}$ by the relation
$$
F^{\mu}_{0}S^{\mu}_{\Lambda}(F^{\mu}_{0})^{*}u^{\mu}_{\lambda}=S^{\Lambda,\mu}_{\lambda}u^{\mu}_{\lambda}
\,.
$$ 
Before stating the next results, let us notice the relations 
\be\label{RR}
\left(-R_{\mu}^{0}+z\right)^{-1}=\frac1z\,\left(\uno+\frac1{z}\,R^{0}_{\mu-\frac1z}\right)\,,
\quad \left(-R_{\mu}^{\Lambda}+z\right)^{-1}=\frac1z\,\left(\uno+\frac1{z}\,R^{\Lambda}_{\mu-\frac1z}\right)\,,
\ee
Therefore, by (H2) and Theorem \ref{LAP}, the limits 
\be\label{RRlim1}
\left(-R_{\mu}^{0}+(\lambda\pm i0)\right)^{-1}:=\lim_{\epsilon\downarrow 0}\left(-R_{\mu}^{0}+(\lambda\pm i\epsilon)\right)^{-1}\,,\quad \lambda\not=0\,,\ \mu-\frac1\lambda\in\Sigma_{0}\,,
\ee 
\be\label{RRlim2}
\left(-R_{\mu}^{\Lambda}+(\lambda\pm i0)\right)^{-1}:=
\lim_{\epsilon\downarrow 0}\left(-R_{\mu}^{\Lambda}+(\lambda\pm i\epsilon)\right)^{-1}\,,
\quad \lambda\not=0\,,\ \mu-\frac1\lambda\in\Sigma_{0}\backslash\Sigma_{\Lambda}\,,
\ee
exist in $\B(L^{2}_{\varphi}(M),L^{2}_{\varphi^{-1}}(M))$.
\begin{theorem}\label{BK}
Suppose that the couple $(A_{0},A_{\Lambda})$ satisfies hypotheses (H1)-(H6). Then the strong limits 
\be\label{WR}
W_{\pm}(R^{\Lambda}_{\mu},R^{0}_{\mu})
:=\text{s-}\lim_{t\to\pm\infty}e^{-itR^{\Lambda}_{\mu}}e^{itR^{0}_{\mu}}P^{\mu}_{ac}
\ee
exist everywhere in $L^{2}(M)$. 
Moreover, for any  $\lambda\not=0$ such that $\mu-\frac1\lambda\in \sigma_{ac}(A_{0})\cap(\Sigma_{0}\backslash\Sigma_{\Lambda})$, one has
\be\label{S1}
S^{\Lambda,\mu}_{\lambda}=\uno-2\pi i\,L^{\mu}_{\lambda}
\Lambda_{\mu}\big(\uno+G^{*}_{\mu}\left(-R_{\mu}^{\Lambda}+(\lambda+ i0)\right)^{-1}G_{\mu}\Lambda_{\mu}\big)
(L^{\mu}_{\lambda})^{*}\,,
\ee
where
\be\label{SR}
L^{\mu}_\lambda: \fh^{*}\to(L^{2}(M)^{0}_{ac})_{\mu-\frac1\lambda}\,,\quad 
L^{\mu}_{\lambda}\phi:=\frac1\lambda(F_{0}G_{\mu}\phi)_{\mu-\frac1\lambda}\,.
\ee
\end{theorem} 
\begin{proof}
By \eqref{resolvent}, one has $R_{\mu}^{\Lambda}-R_{\mu}^{0}=G_{\mu}\Lambda_{\mu}G^{*}_{\mu}$ and we can use \cite[Theorem 4', page 178]{Y} (notice that the maps there denoted by $G$ and $V$ corresponds to our $G_{\mu}^{*}$ and $\Lambda_{\mu}$ respectively). Let us check that the hypotheses there required are satisfied. Since $G^{*}_{\mu}\in \B(L^{2}(M),\fh)$, the operator $G^{*}_{\mu}$ is $|R^{0}_{\mu}|^{1/2}$-bounded. By \eqref{RR}, (H2), Theorem \ref{LAP} and (H6), the limits 
$$
\lim_{\epsilon\downarrow 0}\, G^{*}_{\mu}(-R^{0}_{\mu}+(\lambda\pm i\epsilon))^{-1}\,,
$$
$$
\lim_{\epsilon\downarrow 0}\, G^{*}_{\mu}(-R^{\Lambda}_{\mu}+(\lambda\pm i\epsilon))^{-1}\,,
$$
$$
\lim_{\epsilon\downarrow 0}\, G^{*}_{\mu}(-R^{\Lambda}_{\mu}+(\lambda\pm i\epsilon))^{-1}G_{\mu}
$$
exist. Therefore, to get the thesis we need to check the validity of the remaining hypothesis in 
\cite[Theorem 4', page 178]{Y}: $G^{*}_{\mu}$ is weakly-$R^{0}_{\mu}$ smooth, i.e., by \cite[Lemma 2, page 154]{Y}, 
\be\label{in1.1}
\sup_{0<\epsilon<1}\epsilon\,{\|}G^{*}_{\mu} (-R^{0}_{\mu}+(\lambda\pm i\epsilon))^{-1}{\|}_{L^{2}\!,\fh}^{2}\le c_{\lambda}<+\infty\,,\quad\text{a.e. $\lambda$}\,.
\ee
By \eqref{RR}, this is consequence of
\be\label{in2.2}
\sup_{0<\delta<1}\delta\,{\|}G^{*}_{\mu}R^{0}_{\mu-\frac1\lambda\pm i\delta}{\|}_{L^{2}\!,\fh}^{2}\le C_{\lambda}<+\infty\,,\quad\text{a.e. $\lambda$}\,.
\ee
By 
\begin{align*}
{\|}G^{*}_{\mu} R^{0}_{z}{\|}_{L^{2}\!,\fh}={\|}\tau R^{0}_{\mu}R^{0}_{z}{\|}_{L^{2}\!,\fh}= 
{\|}\tau R^{0}_{z}R^{0}_{\mu}{\|}_{L^{2}\!,\fh}=
{\|}R^{0}_{\mu}(\tau R^{0}_{z})^{*}{\|}_{\fh^{*}\!\!,L^{2}}\le
{\|}R^{0}_{\mu}{\|}_{L^{2}\!,L^{2}} {\|}G_{\bar z}{\|}_{\fh^{*}\!\!,L^{2}}\,,
\end{align*}
\eqref{in2.2} follows by \eqref{eps}, hypotheses (H5.2) and (H6). Thus, by \cite[Theorem 4', page 178]{Y}, 
the limits \eqref{WR} exist everywhere in $L^{2}(M)$ and the corresponding scattering matrix is given by \eqref{S1}, where $L^{\mu}_{\lambda}\phi:=(F^{\mu}_{0}G_{\mu}\phi)_{\lambda}=\frac1\lambda(F_{0}G_{\mu}\phi)_{\mu-\frac1\lambda}$. 
\end{proof}
\begin{lemma}\label{LG} For any $z\not=0$ such that $\mu-\frac1z\in \rho(A_{0})$ one has 
$$
\Lambda_{\mu}\big(\uno+G^{*}_{\mu}\left(-R_{\mu}^{\Lambda}+z\right)^{-1}G_{\mu}\Lambda_{\mu}\big)=\Lambda_{\mu-\frac1z}\,.
$$
\end{lemma}
\begin{proof} 
By \eqref{Gzw}, one has
\begin{align}\label{GG}
G_{\mu}+\frac1{z}\,R^{0}_{\mu-\frac1z}  G_{\mu}=
G_{\mu-\frac1z}\,.
\end{align}
By \eqref{RR}, \eqref{GG} and \eqref{Lambda2}, one obtains
\begin{align*}
&\Lambda_{\mu}+\Lambda_{\mu}G^{*}_{\mu}\left(-R_{\mu}^{\Lambda}+z\right)^{-1}G_{\mu}\Lambda_{\mu}\\
=&
\Lambda_{\mu}+\frac1z\,\Lambda_{\mu}G^{*}_{\mu}\left(\left(G_{\mu}+\frac1{z}\,R^{0}_{\mu-\frac1z}G_{\mu}\right)\Lambda_{\mu}
+G_{\mu-\frac1z}\left(\frac1{z}\,\Lambda_{\mu-\frac1z}G^{*}_{\mu-\frac1{\bar z}}G_{\mu}\Lambda_{\mu}\right)\right)\\
=&
\Lambda_{\mu}+\frac1z\,\Lambda_{\mu}G^{*}_{\mu}G_{\mu-\frac1z}\Lambda_{\mu}
+\frac1z\,\Lambda_{\mu}G^{*}_{\mu}G_{\mu-\frac1z}\big(\Lambda_{\mu-\frac1z}-\Lambda_{\mu}\big)\\
=&
\Lambda_{\mu}+\frac1z\,\Lambda_{\mu}G^{*}_{\mu}G_{\mu-\frac1z}\Lambda_{\mu-\frac1z}
=\Lambda_{\mu-\frac1z}\,.
\end{align*}

\end{proof}
\begin{corollary}\label{S-matrix} Suppose that the couple $(A_{0},A_{\Lambda})$ satisfies hypotheses (H1)-(H6). Then 
\begin{equation*}
S^{\Lambda}_{\lambda}=\uno-2\pi iL_{\lambda}\Lambda^{+}_{\lambda}L_{\lambda}^{*}\,,\quad \lambda\in\sigma_{ac}(A_{0})\cap(\Sigma_{0}\backslash\Sigma_{\Lambda})\,,
\end{equation*}
where  $L_\lambda: \fh^{*}\to(L^{2}(M)^{0}_{ac})_{\lambda}$ is the $\mu$-independent linear operator defined by  
$$
L_{\lambda}\phi:=(\mu-\lambda)(F_{0}G_{\mu}\phi)_{\lambda}\,.
$$ 
\end{corollary}
\begin{proof} By Theorem \ref {AC}, Theorem \ref{BK} and by Birman-Kato invariance principle (see e.g. \cite[Section II.3.3]{BW}), one has
$$
W_{\pm}(A_{\Lambda},A_{0})=W_{\pm}(R^{\Lambda}_{\mu},R^{0}_{\mu})
$$
and so
$$
S_{\Lambda}=S_{\Lambda}^{\mu}\,.
$$
Thus, since $(F^{\mu}_{0}u)_{\lambda}=\frac1\lambda(F_{0}u)_{\mu-\frac1\lambda}$,  one obtains (see also \cite[Equation (14), Section 6, Chapter 2]{Y})
\be\label{SS}
S^\Lambda_{\lambda}=S^{\Lambda,\mu}_{(-\lambda+\mu)^{-1}}\,.
\ee
By Lemma \ref{LG}, whenever $z=\lambda\pm i\epsilon$ and $\mu-\frac1{\lambda}\in\Sigma_{0}\backslash\Sigma_{\Lambda}$, one gets, as $\epsilon\downarrow 0$, 
$$
\Lambda_{\mu}\big(\uno+G^{*}_{\mu}\left(-R_{\mu}^{\Lambda}+(\lambda\pm i0)\right)^{-1}G_{\mu}\Lambda_{\mu}\big)=\Lambda^{\pm}_{\mu-\frac1{\lambda}}\,.
$$
The proof is then concluded by Theorem \ref{BK}, by \eqref{SS} and by setting  
$L_{\lambda}:=L^{\mu}_{{(-\lambda+\mu)^{-1}}}$. The operator $L_{\lambda}$ is $\mu$-independent by invariance principle, let us provide a direct proof: given $\mu_{1}\not=\mu_{2}$, 
by \eqref{Gzw} and by $(F_{0}R^{0}_{\mu}u)_{\lambda}=(-\lambda+\mu)^{-1}(F_{0}u)_{\lambda}$, one gets the identity
\begin{align*}
&\big(L^{\mu_{1}}_{_{(-\lambda+\mu_{1})^{-1}}}-L^{\mu_{2}}_{_{(-\lambda+\mu_{2})^{-1}}}\big)\phi\\
=&(\mu_{1}-\lambda)(F_{0}G_{\mu_{1}}\phi)_{\lambda}-(\mu_{2}-\lambda)(F_{0}G_{\mu_{2}}\phi)_{\lambda}\\
=&(F_{0}((\mu_{1}-\mu_{2})G_{\mu_{1}}-(\lambda-\mu_{2})(G_{\mu_{1}}-G_{\mu_{2}}))\phi)_{\lambda}\\
=&(\mu_{1}-\mu_{2})(F_{0}G_{\mu_{1}}\phi)_{\lambda}-(\lambda-\mu_{2})(\mu_{2}-\mu_{1})(F_{0}R^{0}_{\mu_{2}}G_{\mu_{1}}\phi)_{\lambda}\\
=&0\,.
\end{align*}
\end{proof}
\end{section}
\begin{section}{Applications.} 
Here we take $$A_{0}=\Delta:H^{2}(\RE^{n})\subset L^{2}(\RE^{n})\to L^{2}(\RE^{n})\,,$$ where $H^{s}(\RE^{n})$, $s\in\RE$, denotes the usual scale of Sobolev spaces and where $\Delta $ denotes the distributional Laplacian, and $$\tau:H^{2}(\RE^{n})\to\fh$$  bounded and surjective onto the Hilbert space $\fh$ and such that $\ker(\tau)$ is $L^{2}(\RE^{n})$-dense. \par 
In the following we use the scale of weighted Sobolev spaces $H_{w}^{s}(\RE^{n})$, $s\in\RE$, $w\in\RE$. Here $H_{w}^{0}(\RE^{n})\equiv L_{w}^{2}(\RE^{n})$ denotes the weighted $L^{2}$-space which corresponds, according to the notation in the Section 3, to the choice $\varphi(x)=(1+|x|^{2})^{w/2}$. Then the weighted Sobolev space  $H_{w}^{m}(\RE^{n})$, $m\ge 1$ integer,  consists of functions in $L_{w}^{2}(\RE^{n})$ having $k$-order distributional derivatives, $1\le k\le m$, belonging to $L_{w}^{2}(\RE^{n})$; $H_{w}^{s}(\RE^{n})$, $s>0$ not integer, is defined by interpolation as in the unweighted case and finally $H_{w}^{s}(\RE^{n})$, $s\in(-\infty,0)$, is defined as the dual of  $H_{-w}^{-s}(\RE^{n})$ (see e.g. \cite[Section 4.2]{ET}). 
\begin{theorem}\label{LH} Let $\Delta_{\Lambda}$ denote the self-adjoint extension of the symmetric operator $S:=\Delta|\ker(\tau)$ given in Theorem \ref{teo}, corresponding to the family $\Lambda=\{\Lambda_{z}\}_{z\in Z_{\Lambda}}$, $\Lambda_{z}\in\B(\fb,\fb^{*})$, $\fh\hookrightarrow\fb$, and to  the choice $A_{0}=\Delta$. Suppose that: 
\vskip8pt\noindent 
i) $\Delta_{\Lambda}$ is bounded from above; \par\noindent 
ii) there exists $c_{\Lambda}>0$ such that  the embedding $\ran(\Lambda_{\lambda})\hookrightarrow\fh^{*}$ is compact for any $\lambda>c_{\Lambda}$;
\par\noindent 
iii) there exists $\chi\in C^{\infty}_{comp}(\RE^{n})$ such that, for any $u\in H^{2}(\RE^{n})$,
\be\label{HT}
\tau u=\tau(\chi u)
\ee 
\vskip8pt\noindent
Then  asymptotic completeness holds for the scattering couple  $(\Delta,\Delta_{\Lambda})$, 
$$
\sigma_{ac}(\Delta_{\Lambda})=\sigma_{ess}(\Delta_{\Lambda})=(-\infty,0]\,,\quad 
\sigma_{sc}(\Delta_{\Lambda})=\emptyset 
$$
and the scattering matrix $S_{\lambda}^{\Lambda}$ is given by 
\be\label{S-tau}
S_{\lambda}^{\Lambda}=\uno-2\pi iL_{\lambda}\Lambda_{\lambda}^{+}L^{*}_{\lambda}\,,\quad \lambda\in(-\infty,0)\backslash\sigma_{p}^{-}(\Delta_{\Lambda})\,,
\ee 
where $\sigma_{p}^{-}(\Delta_{\Lambda}):=(-\infty,0)\cap\sigma_{p}(\Delta_{\Lambda})$ is a possibly empty discrete set,  
$$
\Lambda_{\lambda}^{+}:=\lim_{\epsilon\downarrow 0}\Lambda_{\lambda+i\epsilon}\,,\quad\text{the limit existing in $\B(\fb,\fb^{*})$,}
$$
\be\label{Ll}
L_{\lambda}:\fb^{*}\to L^{2}({\mathbb S}^{n-1})\,,\qquad L_{\lambda}\phi(\xi):=\frac1{2^{\frac12}}\,\frac{|\lambda|^{\frac{n-2}4}}{(2\pi)^{\frac{n}2}}\langle \tau(\chi u^{\xi}_{\lambda}),\phi\rangle_{\fb,\fb^{*}}\,.
\ee
Here ${\mathbb S}^{n-1}$ denotes the (n-1)-dimensional unitary sphere in $\RE^{n}$ and $u^{\xi}_{\lambda}$ is the plane wave with direction $\xi\in{\mathbb S}^{n-1}$ and wavenumber $|\lambda|^{\frac12}$, i.e. $u^{\xi}_{\lambda}(x)=e^{i\,|\lambda|^{\frac12}\xi\cdot x}$.
\end{theorem}
\begin{proof}
According to \cite[Lemma 1, page 170]{RS-IV}, one has $R_{z}^{0}\in\B(L_{w}^{2}(\RE^{n}))$ 
; this entails (see \cite[relation (4.8)]{JST}) 
\be\label{(H1)alpha}
(-\Delta+z)^{-1}\equiv R_{z}^{0}\in\B(L_{w}^{2}(\RE^{n}),H_{w}^{2}(\RE^{n}))\,.
\ee 
Thus  $\Delta$ satisfies hypothesis (H1). It is a well-known fact that LAP holds for the free Laplacian, i.e. $\Delta$ satisfies hypothesis (H2) (see e.g. \cite[Theorem 4.1]{Agm}, \cite[Theorem 18.3]{KK}): for any $\lambda<0$ and $w>\frac12$,
\be\label{LAP-Sob}
\text{  $R^{0,\pm}_{\lambda}=\lim_{\epsilon\downarrow 0}\,(-\Delta+\lambda\pm i\epsilon)^{-1}$ exist in $\B(L^{2}_{w}(\RE^{n}),H^{2}_{-w}(\RE^{n}))$}
\ee 
and the maps 
$$
z\mapsto R^{0,\pm}_{z}:=\begin{cases}R^{0}_{z}\,,& z\in \CO\backslash(-\infty,0]\\
R^{0,\pm}_{\lambda}\,,& z=\lambda\in(-\infty,0)\end{cases}
$$
are continuous on $\overline{\CO_{\pm}}\backslash\{0\}$ to $\B(L^{2}_{w}(\RE^{n}),H^{2}_{-w}(\RE^{n}))$. 
Hypothesis (H3) holds true by \cite[Corollary 5.7(b)]{BAD}. Hypothesis (H5) holds true  by \eqref{HT} and \eqref{LAP-Sob}. By \eqref{LAP-Sob}, $\supp(\tau^{*}\phi)\subseteq\supp(\chi)$. Since $G_{z}\phi$ is the convolution of the kernel of $R^{0}_{z}$ with the distribution $\tau^{*}\phi\in H_{comp}^{-2}(\RE^{n})$, one obtains $G_{z}\in\B(\fh^{*}\!\!,L^{2}_{w}(\RE^{n}))$ for any $w$ and hypothesis (H6) holds true. By Remark \ref{R(H6)}, if \eqref{HT} holds then the map $z\mapsto R^{\Lambda}_{z}$ satisfies (H1). If the embedding $\ran(\Lambda_{\lambda})\hookrightarrow\fh^{*}$ is compact, then hypothesis (H4) holds true with $k=1$ by \eqref{resolvent}, by $G_{\lambda}^{*}\in\B(L^{2}(\RE^{n},\fh))$ and  by $G_{\lambda}\in\B(\fh^{*}\!\!,L^{2}_{w}(\RE^{n}))$ for any $w$. Therefore, by Theorem \ref{AC}, asymptotic completeness holds for the scattering couple  $(\Delta,\Delta_{\Lambda})$ and $\sigma_{ac}(\Delta_{\Lambda})=\sigma_{ac}(\Delta)=(-\infty,0]$. 
By Corollary \ref{sc}, $\sigma_{sc}(\Delta_{\Lambda})=\emptyset$. Moreover, since $R^{\Lambda}_{z}-R^{0}_{z}$ is compact by $ii)$ and \eqref{resolvent}, $\sigma_{ess}(\Delta_{\Lambda})=\sigma_{ess}(\Delta)=(-\infty,0]$. \par The scattering matrix $S_{\lambda}^{\Lambda}$ is provided by Corollary \ref{S-matrix}. By \eqref{HT}, the distribution $\tau^{*}\phi\in H^{-2}(\RE^{n})$, $\phi\in\fh^{*}$, is compactly supported, $\supp(\tau^{*}\phi)\subseteq\supp(\chi)$. Setting  $v_{\xi}(x):=\frac{e^{i\,\xi\cdot x}}{(2\pi)^{\frac{n}2}}$, its Fourier transform is given by 
$$
\widehat{\tau^{*}\phi}(\xi)= \langle v_{\xi},{\tau^{*}\phi}\rangle_{H^{2}_{loc}(\RE^{n}),H^{-2}_{comp}(\RE^{n})}= \langle \chi v_{\xi},{\tau^{*}\phi}\rangle_{H^{2}(\RE^{n}),H^{-2}(\RE^{n})}=\langle \tau(\chi v_{\xi}),\phi\rangle_{\fh,\fh^{*}}\,.
$$
The unitary 
map $F_{0}:L^{2}(\RE^{n})\to \int^{\oplus}_{(-\infty,0)}L^{2}({\mathbb S}^{n-1})\,d\lambda\equiv L^{2}((-\infty,0);L^{2}({\mathbb S}^{n-1}))$ given by 
\be\label{fourier}
(F_{0}u)_\lambda(\xi):=-\frac1{2^{1/2}}\,{|\lambda|^{\frac{n-2}4}}\, \widehat u(|\lambda|^{1/2}\xi)\,
\ee 
diagonalizes $A_{0}=\Delta$. Therefore, by  $R^{0}_{\mu}\in \B(H^{-2}(\RE^{n}),L^{2}(\RE^{n}))$ and \eqref{fourier}, one gets 
\begin{align*}
(\mu-\lambda)(F_{0}R_{\mu}\tau^{*}\phi)_{\lambda}(\xi)=-\frac1{2^{1/2}}\,|\lambda|^{\frac{n-2}4}\widehat{\tau^{*}\phi}(|\lambda|^{1/2}\xi)
=
-\frac1{2^{1/2}}\,|\lambda|^{\frac{n-2}4}\langle \tau(\chi v_{|\lambda|^{1/2}\xi}),\phi\rangle_{\fh,\fh^{*}}\,.
\end{align*}
This gives the operator $L_{\lambda}$ provided in Corollary \ref{S-matrix} (notice that for notational convenience in \eqref{Ll} we use $-L_{\lambda}$).\par 
In order to conclude the proof we need to show that the limits $\Lambda^{\pm}_{\lambda}$, which exist in $\B(\fh,\fh^{*})$ by Lemma \ref{bound}, in fact exist in $\B(\fb,\fb^{*})$. By \eqref{Lambda2}, for any $z\in\CO\backslash(-\infty,0]$ one has
\be\label{Lambda-epsilon}
\Lambda_{\lambda\pm i\epsilon}=\Lambda_{z}+(z-(\lambda\pm i\epsilon))\Lambda_{\lambda\pm i\epsilon}\tau R^{0}_{\lambda\pm i\epsilon}G_{z}\Lambda_{z}\,.
\ee
Thus, since $\Lambda_{z}\in \B(\fb,\fb^{*})\subseteq\B(\fb,\fh^{*})$ for any $z\in\CO\backslash(-\infty,0]$, since $\tau R^{0}_{\lambda\pm i\epsilon}$ converges in $\B(L_{w}^{2}(\RE^{n}),\fh)$ by \eqref{HT} and \eqref{LAP-Sob} and since $G_{z}\in\B(\fh^{*},L^{2}_{w}(\RE^{n}))$, the existence in $\B(\fh,\fh^{*})$ of the limits $\Lambda^{\pm}_{\lambda}$ entails 
the existence of such limits in $\B(\fb,\fh^{*})$. Then, by duality and \eqref{Lambda1}, the limits exist in $\B(\fh,\fb^{*})$ as well. Thus, using again \eqref{Lambda-epsilon} and repeating the same reasonings, at the end one gets the existence of the limits $\Lambda^{\pm}_{\lambda}$ in $\B(\fb,\fb^{*})$.
\end{proof}
\begin{subsection}{Traces, layer operators and Dirichlet-to-Neumann maps.}
\begin{subsubsection}{Trace maps and single-layer operators on $d$-sets.} Here we begin recalling some results about $d$-sets which are needed below (see \cite{JW} and \cite{Trib} for more details).\par 
A Borel set $\Gamma\subset\RE^n$ is called a $d$-set, $0<d<n$, if  there exists a Borel
measure $\mu$ in $\RE^n$ such that $\supp(\mu)=\Gamma$ and
\begin{equation}\label{dset}
\exists\, c_{\pm}>0\,:\, \forall x\in \Gamma,\ \forall r\in(0,1),\ 
c_{-}r^d\le\mu(B_{r}^{x}\cap \Gamma)\le c_{+}r^d,
\end{equation}
where $B_{r}^{x}$ is the ball in $\RE^{n}$ of radius $r$ centered at the point
$x$ (see e.g. \cite[Definition 3.1]{Trib}). By \cite[Theorem 3.4]{Trib}, once $\Gamma$ is a $d$-set, $\mu^{d}_\Gamma$, the
$d$-dimensional Hausdorff measure restricted to $\Gamma$, always satisfies \eqref{dset} and so $\Gamma$ has
Hausdorff dimension $d$ in the neighborhood    of any of its
points. \par
Examples of $d$-sets are, whenever $d$ is an integer number, finite unions of $d$-dimensional Lipschitz manifolds 
which intersect on a set of zero $d$-dimensional Hausdorff
measure and, whenever $d$ is not an integer, self-similar fractals of Hausdorff
dimension $d$.   
\par
Let $\gamma^{0}_{\Gamma}$ be  the map defined by the restriction of $u\in C^{\infty}_{comp}(\RE^{n})$ along the set  $\Gamma$: $\gamma^{0}_{\Gamma}u:=u|\Gamma$.  Then, by \cite[Theorem 1, Chapter VII]{JW},  such a map has a bounded and surjective extension to $H^{s+\frac{n-d}2}(\RE^{n})$ for any $s>0$:
\be\label{0trace}
\gamma_{\Gamma}^{0}:H^{s+\frac{n-d}2}(\RE^{n})\to B_{2,2}^{s}(\Gamma)\,.
\ee
Here the Hilbert space $B^{s}_{2,2}(\Gamma)$ is a Besov-like space (see \cite[Section 2, Chapter V]{JW} for the precise definitions). Notice that if $\Gamma$ is a manifold of class $C^{\kappa,1}$, $\kappa\ge 0$, then $B^{s}_{2,2}(\Gamma)=H^{s}(\Gamma)$ for any $s\le\kappa+1$, where $H^{s}(\Gamma)$ denotes the usual fractional Sobolev space on $\Gamma$ (see e.g. \cite[Chapter 3]{McL}). Moreover, in the
case 
$0<s<1$, $B^{s}_{2,2}(\Gamma)$ can
be defined (see \cite[Section 1.1, chap. V]{JW}) as the set of $\phi\in L^2(\Gamma,\mu^{d}_{\Gamma})$ having finite norm
$$\|\phi\|^2_{B^{s}_{2,2}(\Gamma)}:=
\int_{\Gamma}|\phi(x)|^{2}\,d\mu^{d}_{\Gamma}(x)+
\int_{\{(x,y)\in\Gamma\times\Gamma:
|x-y|<1\}}
\frac{|\phi(x)-\phi(y)|^2}{|x-y|^{d+2s}}\,d(\mu^{d}_{\Gamma}\times\mu^{d}_{\Gamma})(x,y)\, .
$$
Since such a norm coincides with the usual norm in $H^{s}(\Gamma)$ whenever $\Gamma$ is a Lipschitz hypersurface, for successive convenience we use the notation $B^{s}_{2,2}(\Gamma)\equiv H^{s}(\Gamma)$ whenever $0<s<1$. We also use the following notations for the dual (with respect to the $L^{2}(\Gamma)$-pairing) spaces: $(B^{s}_{2,2}(\Gamma))^{*}\equiv B^{-s}_{2,2}(\Gamma)$ and, whenever $0<s<1$, $(H^{s}(\Gamma))^{*}\equiv H^{-s}(\Gamma)$.\par
By \cite[Proposition 20.5]{Trib}, one has, similarly to the regular case,
\be\label{emb}
\text{$\Gamma$ compact $d$-set}\quad\Longrightarrow\quad\text{the embedding $B^{s_{2}}_{2,2}(\Gamma)\hookrightarrow B^{s_{1}}_{2,2}(\Gamma)$, $s_{2}>s_{1}$, is compact.}
\ee 
\be\label{embLp}
\text{$\Gamma$ compact $d$-set}\quad\Longrightarrow\quad\text{the embedding $B^{s}_{2,2}(\Gamma)\hookrightarrow L^{\frac{2d}{d-2s}}(\Gamma)$,  $0<2s<d$, is compact.}
\ee 
In the following the resolvent $R^{0}_{z}\equiv(-\Delta+z)^{-1}$, $z\in\CO\backslash(-\infty,0]$, is viewed as a map in $\B(H^{s}(\RE^{n}), H^{s+2}(\RE^{n}))$, $s\in\RE$. Given $s>0$, by the mapping properties  \eqref{0trace} one gets, for the adjoint of the trace map, 
$$(\gamma_{\Gamma}^{0})^{*}:B_{2,2}^{-s}(\Gamma)\to H^{-s-\frac{n-d}2}(\RE^{n})$$ and so we can define the bounded operator (the single-layer potential)
\be\label{SL}
\SL_{z}:=R^{0}_{z}(\gamma_{\Gamma}^{0})^{*}: B_{2,2}^{-s}(\Gamma)\to H^{2-s-\frac{n-d}2}(\RE^{n})
\,.
\ee
Notice that $\SL_{z}=G_{z}$ whenever $\tau=\gamma_{\Gamma}^{0}$ and $s=2-\frac{n-d}2$. By resolvent identity one has (compare with \eqref{Gwz})
\be\label{Swz}
\SL_{z}=\SL_{w}+(w-z)R^{0}_{z}\SL_{w}\,.
\ee
If $n-2<d<n$, by \eqref{0trace} and \eqref{SL}, one obtains  the bounded operator
$$
\gamma_{\Gamma}^{0}\SL_{z}: B_{2,2}^{-s}(\Gamma)\to B_{2,2}^{2-s-(n-d)}(\Gamma)\,.
$$
Since the map $z\mapsto R^{0}_{z}$ is analytic on $\CO\backslash(-\infty,0]$ to $\B(H^{s}(\RE^{n}),H^{s+2}(\RE^{n}))$ for any $s\in\RE$, the maps $z\mapsto \SL_{z}$ and $z\mapsto \gamma_{\Gamma}^{0}\SL_{z}$ are analytic as well, with values in $\B(B_{2,2}^{-s}(\Gamma),H^{2-s-\frac{n-d}2}(\RE^{n}))$ and $\B(B_{2,2}^{-s}(\Gamma), B_{2,2}^{2-s-(n-d)}(\Gamma))$ respectively .
\par
By \eqref{LAP-Sob}, duality and interpolation one gets 
\be\label{interp}
R^{0,\pm}_{\lambda}\in \B(H^{-s}_{w}(\RE^{n}),H^{-s+2}_{-w}(\RE^{n}))\,,\qquad 0\le s\le 2\,.
\ee 
Thus, since $\Gamma$ is bounded, 
the limits $$
\SL^{\pm}_{\lambda}:=R^{0,\pm}_{\lambda}(\gamma_{\Gamma}^{0})^{*}=\lim_{\epsilon\downarrow 0}(\gamma_{\Gamma}^{0}R^{0}_{\lambda\mp i\epsilon})^{*}
$$
exist in $\B(B_{2,2}^{-s}(\Gamma), H^{2-s-\frac{n-d}2}_{-w}(\RE^{n}))$, $0<s\le 2-\frac{n-d}2$. 
\end{subsubsection}
\begin{subsubsection}{Single- and double-layer operators on Lipschitz boundaries} Let $\Gamma$ be the boundary of a bounded Lipschitz domain $\Omega$; we set $\Omega_{\-}\equiv\Omega$ and $\Omega_{\+}:=\RE^{n}\backslash\overline\Omega$.  In the following $\Delta_{\Omega_{\-/\+}}$ denote the distributional Laplacians on $\Omega_{\-/\+}$.
\par 
The one-sided, zero and first order, trace operators $\gamma_{\Gamma}^{0,\-/\+}$ and $\gamma_{\Gamma}^{1,\-/\+}=\nu_{\Gamma} \cdot\gamma_{\Gamma}^{0,\-/\+}\nabla$ ($\nu_{\Gamma} $ denoting the outward normal vector at the boundary) defined on
smooth functions in $\mathcal{C}_{comp}^{\infty}(  \overline{\Omega}_{\-/\+
})  $ extend to bounded and surjective linear operators (see e.g. \cite[Theorem
3.38]{McL})
\begin{equation}
\gamma_{\Gamma}^{0,\-/\+}\in{\B}(H^{s+1/2}(  \Omega_{\-/\+})  ,H^{s}(  \Gamma)  )\,,\qquad 0<s< 1\,.
\label{Trace_Gamma_plusmin_est}%
\end{equation}
and
\begin{equation}
\gamma_{\Gamma}^{1,\-/\+}\in{\B}(H^{s+3/2}(  \Omega_{\-/\+})  ,H^{s}(  \Gamma)  )\,,\qquad 0<s<1
\label{Trace_Gamma_plusmin_est_1}%
\end{equation}
(we refer to \cite[Chapter 3]{McL} for the definition of the Sobolev spaces $H^{s}(\Omega_{\-/\+})$ and $H^{s}(\Gamma)$). Using these maps and setting $H^{s}(  \RE^n\backslash\Gamma ):=H^{s}(  \Omega_{\-})  \oplus H^{s}(  \Omega
_{\+})$, the two-sided bounded and surjective trace operators are defined according
to%
\begin{equation}
\gamma_{\Gamma}^{0}:H^{s+1/2}(  \RE^n\backslash\Gamma )\rightarrow H^{s}(\Gamma)\,,\quad\gamma_{\Gamma}^{0}%
(u_{\-}\oplus u_{\+}):=\frac{1}{2}(\gamma_{\Gamma}^{0,\-}u_{\-}+\gamma_{\Gamma}^{0,\+}u_{\+})\,,
\label{trace_ext_0}%
\end{equation}%
\begin{equation}
\gamma_{\Gamma}^{1}:H^{s+3/2}(  \RE^n\backslash\Gamma )\rightarrow H^{s}(\Gamma)\,,\quad\gamma_{\Gamma}^{1}%
(u_{\-}\oplus u_{\+}):=\frac{1}{2}(\gamma_{\Gamma}^{1,\-}u_{\-}+\gamma_{\Gamma}^{1,\+}u_{\+})\,,
\label{trace_ext_1}%
\end{equation}
while the corresponding jumps are%
\begin{equation}
[\gamma_{\Gamma}^{0}]:H^{s+1/2}(  \RE^n\backslash\Gamma )  \rightarrow H^{s}(\Gamma)\,,\quad[
\gamma_{\Gamma}^{0}](u_{\-}\oplus u_{\+}):=\gamma_{\Gamma}^{0,\-}u_{\-}-\gamma_{\Gamma}^{0,\+}u_{\+}\,,
\end{equation}%
\begin{equation}
[\gamma_{\Gamma}^{1}]:H^{s+3/2}(  \RE^n\backslash\Gamma )\rightarrow H^{s}(\Gamma)\,,\quad[
\gamma_{\Gamma}^{1}](u_{\-}\oplus u_{\+}):=\gamma_{\Gamma}^{1,\-}u_{\-}-\gamma_{\Gamma}^{1,\+}u_{\+}\,.
\end{equation}
Let us notice that in the case $u=u_{\-}\oplus u_{\+}\in H^{s+1/2}(\RE^{n})$, $0<s< 1$, $\gamma^{0}_{\Gamma}$ in \eqref{trace_ext_0} coincides with the map defined in \eqref{0trace} and so there is no ambiguity in our notations; this also entails that $\gamma^{0}_{\Gamma}$ remains  surjective even if restricted to $H^{2}(\RE^{n})$. Similarly the map $\gamma^{1}_{\Gamma}$ is surjective onto $H^{s}(\Gamma)$ even if restricted to $H^{s+3/2}(\RE^{n})$.\par
By \cite[Lemma 4.3]{McL}, the trace maps $\gamma_{\Gamma}^{1,\-/\+}$ can be extended to the spaces $$H^{1}_{\Delta}(\Omega_{\-/\+}):=\{u_{\-/\+}\in H^{1}(\Omega_{\-/\+}):\Delta_{\Omega_{\-/\+}}u_{\-/\+}\in L^{2}(\Omega_{\-/\+})\}\,:$$ 
$$
\gamma_{\Gamma}^{1,\-/\+}: H^{1}_{\Delta}(\Omega_{\-/\+})\to H^{-1/2}(\Gamma)\,.
$$
This gives the analogous extensions of the maps $\gamma_{\Gamma}^{1}$ and $[\gamma_{\Gamma}^{1}]$ defined on $H^{1}_{\Delta}(\RE^{n}\backslash\Gamma):=H^{1}_{\Delta}(\Omega_{\-})\oplus H^{1}_{\Delta}(\Omega_{\+})$ with values in $H^{-1/2}(\Gamma)$. \par
By using a cut-off function $\chi\in \mathcal{C}_{comp}^{\infty}( \RE^{n})$ such that $\chi=1$ in a neighborhood    of $\Omega_{\-}$, all the maps defined above can be extended (and we use the same notation) to functions $u$ such that $\chi u$ is in the right function space.\par
The single-layer operator $\SL_{z}$ has been already introduced in the previous subsection, see \eqref{SL}; here we recall the definition of double-layer operator $\DL_{z}$, $z\in\CO\backslash(-\infty,0]$: by the dual map
$$
(\gamma_{\Gamma}^{1})^{*}:H^{-s}(\Gamma)\to H^{-s-3/2}(\RE^{n})
$$
and by the resolvent $R^{0}_{z}\in\B(H^{s}(\RE^{n}),H^{s+2}(\RE^{n}))$, one defines the bounded operator
\be\label{DL}
\DL_{z}:H^{-s}(\Gamma)\to H^{-s+1/2}(\RE^{n})\,,\quad \DL_{z}:=R^{0}_{z}(\gamma_{\Gamma}^{1})^{*}\,,\quad 0<s< 1\,.
\ee
Let us notice that $\DL_{z}=G_{z}$ whenever $\tau=\gamma_{\Gamma}^{1}$ and $s=\frac12$.
By resolvent identity one has (compare with \eqref{Gwz})
\be\label{Dwz}
\DL_{z}=\DL_{w}+(z-w)R^{0}_{z}\DL_{w}\,.
\ee
By the mapping properties of the layer operators, one gets (see \cite[Theorem 6.11]{McL})
\be\label{map}
\chi\SL_{z}\in \B(H^{-1/2}(\Gamma),H^{1}(\RE^{n}))\,,\qquad\chi\DL_{z}\in \B(H^{1/2}(\Gamma),H^{1}(\RE^{n}\backslash\Gamma))\,,
\ee
for any $\chi\in C^{\infty}_{comp}(\RE^{n})$; by $(-(\Delta_{\Omega_{\-}}\oplus\Delta_{\Omega_{\+}})+z)\SL_{z}\phi=(-(\Delta_{\Omega_{\-}}\oplus\Delta_{\Omega_{\+}})+z)\DL_{z}\varphi=0$, one gets $\chi\SL_{z}\phi\in H^{1}_{\Delta}(\RE^{n}\backslash\Gamma)$, $\phi\in H^{1/2}(\Gamma)$, and $\chi\DL_{z}\varphi\in H^{1}_{\Delta}(\RE^{n}\backslash\Gamma)$, $\varphi\in H^{-1/2}(\Gamma)$. Thus
$$
\gamma_{\Gamma}^{0}\SL_{z}\in \B(H^{-1/2}(\Gamma),H^{1/2}(\Gamma))\,,\qquad 
\gamma_{\Gamma}^{1}\DL_{z}\in \B(H^{1/2}(\Gamma),H^{-1/2}(\Gamma))\,.
$$
These mapping properties can be extended to a larger range of Sobolev spaces (see \cite[Theorem 6.12 and successive remarks]{McL}): 
$$
\chi\SL_{z}\in \B(H^{s-1/2}(\Gamma),H^{s+1}(\RE^{n}))\,,\quad 
\chi\DL_{z}\in \B(H^{s+1/2}(\Gamma),H^{s+1}(\RE^{n}\backslash\Gamma))\,,\quad
-1/2\le s\le 1/2\,, 
$$
$$
\gamma_{\Gamma}^{0}\SL_{z}\in \B(H^{s-1/2}(\Gamma),H^{s+1/2}(\Gamma))\,,\quad 
\gamma_{\Gamma}^{1}\DL_{z}\in \B(H^{s+1/2}(\Gamma),H^{s-1/2}(\Gamma))\,,\quad 
-1/2\le s\le 1/2\,
$$
and the following jump relations holds (see e.g. \cite[Theorem 6.11]{McL})
\be\label{jump}
[\gamma_{\Gamma}^{0}]\SL_{z}\phi=0\,,\quad [\gamma_{\Gamma}^{1}]\SL_{z}\phi=-\phi\,,
\quad 
[\gamma_{\Gamma}^{0}]\DL_{z}\varphi=\varphi\,,\quad [\gamma_{\Gamma}^{1}]\DL_{z}\varphi=0\,.
\ee
Since the map $z\mapsto R^{0}_{z}$ is analytic on $\CO\backslash(-\infty,0]$ to $\B(H^{s}(\RE^{n}),H^{s+2}(\RE^{n}))$, the maps $z\mapsto \gamma_{\Gamma}^{0}\SL_{z}$, $z\mapsto \gamma_{\Gamma}^{1}\DL_{z}$, are analytic as well.
\par
By \eqref{interp}, since $\Gamma$ is bounded, 
the limits 
$$
\SL^{\pm}_{\lambda}:=R^{0,\pm}_{\lambda}(\gamma_{\Gamma}^{0})^{*}=\lim_{\epsilon\downarrow 0}\SL_{\lambda\pm i\epsilon}
\,,\qquad
\DL^{\pm}_{\lambda}:=R^{0,\pm}_{\lambda}(\gamma_{\Gamma}^{1})^{*}=\lim_{\epsilon\downarrow 0}\DL_{\lambda\pm i\epsilon}
$$
exist in $\B(B_{2,2}^{-s}(\Gamma), H^{3/2-s}_{-w}(\RE^{n}))$, $0<s\le 3/2$, and
$\B(H^{-s}(\Gamma), H^{1/2-s}_{-w}(\RE^{n}))$, $0<s\le 1/2$, respectively. Moreover, by the identities \eqref{Swz},\eqref{Dwz} and by $\SL_{z}\in \B(B_{2,2}^{-3/2}(\Gamma), L^{2}_{w}(\RE^{n}))$, $\DL_{z}\in \B(H^{-1/2}(\Gamma), L^{2}_{w}(\RE^{n}))$ (see \cite[relation (4.10)]{JST}) one has
\be\label{FR}
\SL^{\pm}_{\lambda}=\SL_{z}+(z-\lambda)R^{0,\pm}_{\lambda}\SL_{z}\,,\quad 
\DL^{\pm}_{\lambda}=\DL_{z}+(z-\lambda)R^{0,\pm}_{\lambda}\DL_{z}\,.
\ee
This entails, for any $ -1/2\le s\le 1/2$,
\be\label{map-LAP}
\chi\SL^{\pm}_{\lambda}\in \B(H^{s-1/2}(\Gamma),H^{s+1}(\RE^{n}))\,,\quad 
\chi\DL^{\pm}_{\lambda}\in \B(H^{s+1/2}(\Gamma),H^{s+1}(\RE^{n}\backslash\Gamma))\,,
\ee
\be\label{0-1-LAP}
\gamma_{\Gamma}^{0}\SL^{\pm}_{\lambda}\in \B(H^{s-1/2}(\Gamma),H^{s+1/2}(\Gamma))\,,\quad 
\gamma_{\Gamma}^{1}\DL^{\pm}_{\lambda}\in \B(H^{s+1/2}(\Gamma),H^{s-1/2}(\Gamma))\,,
\ee
and, by \eqref{jump} and \eqref{FR},
\be\label{jump-LAP}
[\gamma_{\Gamma}^{0}]\SL^{\pm}_{\lambda}\phi=0\,,\quad [\gamma_{\Gamma}^{1}]\SL^{\pm}_{\lambda}\phi=-\phi\,,
\quad 
[\gamma_{\Gamma}^{0}]\DL^{\pm}_{\lambda}\varphi=\varphi\,,\quad [\gamma_{\Gamma}^{1}]\DL^{\pm}_{\lambda}\varphi=0\,.
\ee
Since the maps $z\mapsto R^{0,\pm}_{z}$ are continuous on $\overline{\CO_{\pm}}\backslash\{0\}$ to $\B(L^{2}_{w}(\RE^{n}), H^{2}_{-w}(\RE^{n}))$, the maps 
$$
z\mapsto \gamma_{\Gamma}^{0}\SL_{z}^{\pm}:=\begin{cases}\gamma_{\Gamma}^{0}\SL_{z}\,,& z\in\CO\backslash(-\infty,0]\\
\gamma_{\Gamma}^{0}\SL_{\lambda}^{\pm}\,,&z=\lambda\in(-\infty,0)\,,\end{cases}
$$
$$
z\mapsto \gamma_{\Gamma}^{1}\DL_{z}^{\pm}:=\begin{cases}\gamma_{\Gamma}^{1}\DL_{z}\,,& z\in\CO\backslash(-\infty,0]\\
\gamma_{\Gamma}^{1}\DL_{\lambda}^{\pm}\,,&z=\lambda\in(-\infty,0)\,,\end{cases}
$$
are 
continuous as well. 
\end{subsubsection}
\begin{subsubsection}{The Dirichlet-to-Neumann and Neumann-to-Dirichlet operators.} Let $\Omega\subset\RE^{n}$ be a bounded Lipschitz domain and let us consider the boundary value problems (here $\Omega_{\-}\equiv\Omega$ and $\Omega_{\+}:=\RE^{n}\backslash\overline\Omega$ as in the previous subsection)  
\begin{equation}
\left\{
\begin{array}
[c]{l}%
(-\Delta_{\Omega_{\-}}+z)  u^{z,\-}_{\phi}=0\,,\quad z\in\rho(\Delta^{D}_{\Omega_{\-}})
\\
\gamma_{\Gamma}^{0,\-}u^{z,\-}_{\phi}=\phi\in H^{1/2}(\Gamma)
\end{array}
\right.\quad \left\{
\begin{array}
[c]{l}%
(-\Delta_{\Omega_{\-}}+z)  v^{z,\-}_{\varphi}=0\,,\quad z\in\rho(\Delta^{N}_{\Omega_{\-}})
\\
\gamma_{\Gamma}^{1,\-}v^{z,\-}_{\varphi}=\varphi\in H^{-1/2}(\Gamma)
\end{array}
\right.    \label{in}%
\end{equation}
and
\begin{equation}
\left\{
\begin{array}
[c]{l}%
(-\Delta_{\Omega_{\+}}+z)  u^{z,\+}_{\phi}=0\,,\quad z\in\rho(\Delta^{D}_{\Omega_{\+}})
\\
\gamma_{\Gamma}^{0,\+}u^{z,\+}_{\phi}=\phi\in H^{1/2}(\Gamma)
\\
\text{$u^{z,\+}_{\phi}$ radiating}
\end{array}
\right.  
\quad \left\{
\begin{array}
[c]{l}%
(-\Delta_{\Omega_{\+}}+z)  v^{z,\+}_{\varphi}=0\,,\quad z\in\rho(\Delta^{N}_{\Omega_{\+}})
\\
\gamma_{\Gamma}^{1,\+}v^{z,\+}_{\varphi}=\varphi\in H^{-1/2}(\Gamma)
\\
\text{$v^{z,\+}_{\varphi}$ radiating}
\end{array}
\right. 
\label{ex}%
\end{equation}
where $\Delta^{D}_{\Omega_{\-/\+}}$ and $\Delta^{N}_{\Omega_{\-/\+}}$ denote the Dirichlet and Neumann Laplacian respectively;  we refer to \cite[Definition 9.5]{McL} for the definition of radiating solutions in the exterior problem. 
By \cite[Theorem 4.10(i)]{McL},  the solutions $u^{z,\-}_{\phi}$ and $v_{\varphi}^{z,\-}$ of \eqref{in} exist and are unique in $H_{\Delta}^{1}(\Omega_{\-})$; by \cite[Theorem 9.11 and Exercise 9.5]{McL} the solutions $u^{z,\+}_{\phi}$ and $v^{z,\+}_{\varphi}$ of \eqref{ex} exist and are unique in $H_{\Delta,loc}^{1}(\Omega_{\+}):=\{u: u|\Omega_{\+}\cap B\in H_{\Delta}^{1}(\Omega_{\+}\cap B)\ \text{for any open ball $B\supset\overline\Omega$} \}$. Therefore the Dirichlet-to-Neumann and Neumann-to-Dirichlet operators
$$
{P}_{z}^{\-/\+}:H^{1/2}(\Gamma)\to H^{-1/2}(\Gamma)\,,\quad z\in\rho(\Delta^{D}_{\Omega_{\-}})\,,\qquad {P}_{z}^{\-/\+}\phi:=\gamma_{\Gamma}^{1,\-/\+}u^{z,\-/\+}_{\phi}\,,$$
$$
{Q}_{z}^{\-/\+}:H^{-1/2}(\Gamma)\to H^{1/2}(\Gamma)\,,\quad z\in\rho(\Delta^{N}_{\Omega_{\-}})\,,\qquad {Q}_{z}^{\-/\+}\varphi:=\gamma_{\Gamma}^{0,\-/\+}v^{z,\-/\+}_{\varphi}
$$
are well-defined.\par
Let $\tilde\phi\in H^{-1/2}(\Gamma)$ and  $\tilde\varphi\in H^{1/2}(\Gamma)$; the functions $\SL^{+}_{z}\tilde\phi_{z}|\Omega_{\-/\+}$ and $\DL^{+}_{z}\tilde\varphi_{z}|\Omega_{\-/\+}$ solve \eqref{in} and \eqref{ex} with $\phi=\gamma_{\Gamma}^{0}\SL^{+}_{z}\tilde\phi$ and $\varphi=\gamma_{\Gamma}^{1}\DL^{+}_{z}\tilde\varphi$ (they are radiating according to \cite[Lemma 5.3]{JST}). By \eqref{jump} and \eqref{jump-LAP},
$$
({P}_{z}^{\+}-{P}_{z}^{\-})\gamma_{\Gamma}^{0}\SL^{+}_{z}\tilde\phi=
\gamma_{\Gamma}^{1,\+}(\SL^{+}_{z}\tilde\phi|\Omega_{\+}) -
\gamma_{\Gamma}^{1,\-}(\SL^{+}_{z}\tilde\phi|\Omega_{\-})=[\gamma_{\Gamma}^{1}]\SL^{+}_{z}\tilde\phi=-\tilde\phi
\,,$$
$$
({Q}_{z}^{\+}-{Q}_{z}^{\-})\gamma_{\Gamma}^{1}\DL^{+}_{z}\tilde\varphi=
\gamma_{\Gamma}^{0,\+}(\DL^{+}_{z}\tilde\varphi|\Omega_{\+}) -
\gamma_{\Gamma}^{0,\-}(\DL^{+}_{z}\tilde\varphi|\Omega_{\-})=[\gamma_{\Gamma}^{0}]\DL^{+}_{z}\tilde\varphi=\tilde\varphi\,.
$$
This shows that 
\be\label{ker}
\ker(\gamma_{\Gamma}^{0}\SL^{+}_{z})=\ker(\gamma_{\Gamma}^{1}\DL^{+}_{z})=\{0\}\,.
\ee
By 
\eqref{FR},
one has  
$$
\gamma_{\Gamma}^{0}\SL^{\pm}_{\lambda}=\gamma_{\Gamma}^{0}\SL_{z}+(z-\lambda)\gamma_{\Gamma}^{0}R^{0,\pm}_{\lambda}\SL_{z}\,,\quad 
\gamma_{\Gamma}^{1}\DL^{\pm}_{\lambda}=\gamma_{\Gamma}^{1}\DL_{z}+(z-\lambda)\gamma_{\Gamma}^{1}R^{0,\pm}_{\lambda}\DL_{z}\,.
$$
Since $\ran(\gamma_{\Gamma}^{0}R^{0,\pm}_{\lambda})\subseteq B^{3/2}_{2,2}(\Gamma)$ and $\ran(\gamma_{\Gamma}^{1}R^{0,\pm}_{\lambda})\subseteq H^{1/2}(\Gamma)$, by the compact embeddings \eqref{emb}, one gets  
$$\gamma_{\Gamma}^{0}\SL^{\pm}_{\lambda}-\gamma_{\Gamma}^{0}\SL_{z}\in {\mathfrak S}_{\infty}(H^{-1/2}(\Gamma),H^{1/2}(\Gamma))$$ 
and  
$$\gamma_{\Gamma}^{1}\DL^{\pm}_{\lambda}-\gamma_{\Gamma}^{1}\DL_{z}\in {\mathfrak S}_{\infty}(H^{1/2}(\Gamma),H^{-1/2}(\Gamma))\,.
$$ 
By \cite[Theorems 7.6 and 7.8]{McL}), both $\gamma_{\Gamma}^{0}\SL_{z}$ and $\gamma_{\Gamma}^{1}\DL_{z}$ are Fredholm with zero index; therefore both $\gamma_{\Gamma}^{0}\SL^{\pm}_{z}$ and $\gamma_{\Gamma}^{1}\DL^{\pm}_{z}$ are Fredholm with zero index as well. Thus, by \eqref{ker}, both $\gamma_{\Gamma}^{0}\SL^{+}_{z}$ and $\gamma_{\Gamma}^{1}\DL^{+}_{z}$ have bounded inverses and 
\begin{align}
(\gamma_{\Gamma}^{0}\SL^{+}_{z})^{-1}=&{P}_{z}^{\-}-{P}_{z}^{\+}\,,\quad z\in\overline{\CO_{+}}\backslash\big(\sigma_{disc}(\Delta^{D}_{\Omega_{\-}})\cup\sigma_{disc}(\Delta^{D}_{\Omega_{\+}})\big)\cup\big(\CO\backslash(-\infty,0]\big)\,,\label{P-Q-1}
\\
(\gamma_{\Gamma}^{1}\DL^{+}_{z})^{-1}=&{Q}_{z}^{\+}-{Q}_{z}^{\-}\,,
\quad z\in \overline{\CO_{+}}\backslash\big(\sigma_{disc}(\Delta^{N}_{\Omega_{\-}})\cup\sigma_{disc}(\Delta^{N}_{\Omega_{\+}})\big)\cup\big(\CO\backslash(-\infty,0]\big)\,.\label{P-Q-2}
\end{align}
By the mapping properties of the layer operators, for all $s\in
[0,1/2]  $ the maps $z\mapsto {P}_{z}^{\+}-{P}_{z}^{\-}$ and $z\mapsto {Q}_{z}^{\+}-{Q}_{z}^{\-}$
are analytic on $\CO\backslash(-\infty,0]$ to ${\B}(  H^{s+1/2}(  \Gamma)  ,H^{s-1/2}(
\Gamma))$ and to ${\B}(  H^{s-1/2}(  \Gamma)  ,H^{s+1/2}(
\Gamma)) $ respectively.
\end{subsubsection}
\begin{subsubsection}{The minimal and maximal Laplacian on Lipschitz domains.}
Let $\Delta^{\circ}_{\Omega_{\-/\+}}$ denote the Laplacian in $L^{2}(\Omega_{\-/\+})$ with 
domain $\dom(\Delta^{\circ}_{\Omega_{\-/\+}})= C^{\infty}_{comp}(\Omega_{\-/\+})$. It is immediate to check (see e.g. \cite[Section 2.3]{Leis}) that $(\Delta^{\circ}_{\Omega_{\-/\+}})^{*}=\Delta^{\max}_{\Omega_{\-/\+}}$, where $\Delta^{\max}_{\Omega_{\-/\+}}$ denotes the distributional Laplacian  with domain $$\dom(\Delta^{\max}_{\Omega_{\-/\+}})=H^{0}_{\Delta}(\Omega_{\-/\+}):=\{u_{\-/\+}\in L^{2}(\Omega_{\-/\+}):\Delta_{\Omega_{\-/\+}}u_{\-/\+}\in L^{2}(\Omega_{\-/\+})\}\,.
$$ 
Moreover $\Delta^{\circ}_{\Omega_{\-/\+}}$ is closable with closure given by $\overline{ \Delta^{\circ}_{\Omega_{\-/\+}}}=\Delta^{\min}_{\Omega_{\-/\+}}$ (see \cite[Section 2.3]{Leis}), where $\Delta^{\min}_{\Omega_{\-/\+}}$ denotes the distributional Laplacian with domain $\dom(\Delta^{\min}_{\Omega_{\-/\+}})=H^{2}_{0}(\Omega_{\-/\+})$ and $H^{2}_{0}(\Omega_{\-/\+})$ denotes as usual the completion of $C^{\infty}_{comp}(\Omega_{\-/\+})$ with respect to the $H^{2}$-norm. Therefore 
\be\label{min-max}
(\Delta^{\min}_{\Omega_{\-/\+}})^{*}=\Delta^{\max}_{\Omega_{\-/\+}}\,.
\ee
Since  $$H^{2}_{0}(\Omega_{\-/\+})=\{u_{\-/\+}\in H^{2}(\Omega_{\-/\+}):\gamma_{\Gamma}^{0,\+/\-}u_{\-/\+}=\gamma_{\Gamma}^{1,\+/\-}u_{\-/\+}=0\}$$ (see \cite[Theorem 1]{Mar}) and $H^{2}(\RE^{n})=(H^{2}(\Omega_{\-})\oplus H^{2}(\Omega_{\+}))\cap\ker([\gamma_{\Gamma}^{0}])\cap \ker([\gamma_{\Gamma}^{1}])$ (see e.g. \cite[Theorem 3.5.1]{Agra}), one has 
\be\label{minimal}
\Delta^{\min}_{\Omega_{\-}}\oplus\Delta^{\min}_{\Omega_{\+}}=\Delta|\ker(\tau)\,,\qquad\tau=\gamma_{\Gamma}^{0}\oplus\gamma_{\Gamma}^{1}\,.
\ee
Notice that for a generic Lipschitz boundary $\ran(\tau)$ is strictly contained in $B^{3/2}_{2,2}(\Gamma)\oplus H^{1/2}(\Gamma)$ (see \cite[Corollary 7.11]{MMS}), while $\ran(\tau)=H^{3/2}(\Gamma)\oplus H^{1/2}(\Gamma)$ whenever 
$\Gamma$ is of class $C^{\kappa,1}$, $\kappa>\frac12$ (see  \cite[Theorem 2]{Mar}).\par  
By Green's formula (see e.g. \cite[Theorem 4.4]{McL}), for any couple  $u_{\-/\+}$, $v_{\-/\+}$ in $H^{1}_{\Delta}(\Omega_{\-/\+})$  one has
\begin{align*}
&\langle\Delta_{\Omega_{\-/\+}}u_{\-/\+},v_{\-/\+}\rangle_{L^{2}(\Omega_{\-/\+})}-
\langle u_{\-/\+},\Delta_{\Omega_{\-/\+}}v_{\-/\+}\rangle_{L^{2}(\Omega_{\-/\+})}\\
=&j_{\-/\+}\langle \gamma_{\Gamma}^{0,\-/\+}u_{\-/\+},\gamma_{\Gamma}^{1,\-/\+}v_{\-/\+}\rangle_{H^{1/2}(\Gamma),H^{-1/2}(\Gamma)}\\
-&j_{\-/\+}\langle \gamma_{\Gamma}^{1,\-/\+}u_{\-/\+},\gamma_{\Gamma}^{0,\-/\+}v_{\-/\+}\rangle_{H^{-1/2}(\Gamma),H^{1/2}(\Gamma)}\,,
\end{align*}
where $j_{\-}=-1$ and $j_{\+}=+1$. Therefore, for any couple $u=u_{\-}\oplus u_{\+}$, $v=v_{\-}\oplus v_{\+}$ in $H^{1}_{\Delta}(\RE^{n}\backslash\Gamma)\cap \ker([\gamma_{\Gamma}^{0}])= H^{1}(\RE^{n})\cap H^{0}_{\Delta}(\RE^{n}\backslash\Gamma)$ one has  
\be\begin{split}
&\langle(\Delta_{\Omega_{\-}}\oplus \Delta_{\Omega_{\+}})u,v\rangle_{L^{2}(\RE^{n})}-
\langle u ,(\Delta_{\Omega_{\-}}\oplus \Delta_{\Omega_{\+}})v\rangle_{L^{2}(\RE^{n})}
\\
=&\langle \gamma_{\Gamma}^{0}u,[\gamma_{\Gamma}^{1}]v\rangle_{H^{1/2}(\Gamma),H^{-1/2}(\Gamma)}
-\langle [\gamma_{\Gamma}^{1}]u,\gamma_{\Gamma}^{0}v\rangle_{H^{-1/2}(\Gamma),H^{1/2}(\Gamma)}
\end{split}\label{Green0}\ee
and, for any couple $u=u_{\-}\oplus u_{\+}$, $v=v_{\-}\oplus v_{\+}$ in $H^{1}_{\Delta}(\RE^{n}\backslash\Gamma)\cap \ker([\gamma_{\Gamma}^{1}])$ one has  
\be\begin{split}\label{Green1}
&\langle(\Delta_{\Omega_{\-}}\oplus \Delta_{\Omega_{\+}})u,v\rangle_{L^{2}(\RE^{n})}-
\langle u ,(\Delta_{\Omega_{\-}}\oplus \Delta_{\Omega_{\+}})v\rangle_{L^{2}(\RE^{n})}
\\
=&\langle [\gamma_{\Gamma}^{0}]u,\gamma_{\Gamma}^{1}v\rangle_{H^{1/2}(\Gamma),H^{-1/2}(\Gamma)}
-\langle \gamma_{\Gamma}^{1}u,[\gamma_{\Gamma}^{0}]v\rangle_{H^{-1/2}(\Gamma),H^{1/2}(\Gamma)}\,.
\end{split}\ee
\end{subsubsection}
\end{subsection}
\begin{subsection}{The Laplace operator with Dirichlet boundary conditions on Lipschitz domains.}\label{dirichlet} 
Here we apply Theorem \ref{LH} to a case in which $\tau=\gamma^{0}_{\Gamma}$, $\fh=B^{3/2}_{2,2}(\Gamma)$, $\fb=H^{1/2}(\Gamma)$ and $\Gamma$ is the Lipschitz boundary of a bounded open set $\Omega\subset\RE^{n}$.\par 
Let $$\Delta^{D}_{\Omega_{\-/\+}}:\dom(\Delta^{D}_{\Omega_{\-/\+}})\subset L^{2}(\Omega_{\-/\+})\to L^{2}(\Omega_{\-/\+})\,,\qquad \Delta^{D}_{\Omega_{\-/\+}}u:=\Delta_{\Omega_{\-/\+}} u\,,
$$ 
$$\dom(\Delta^{D}_{\Omega_{\-/\+}}):=\{u_{\-/\+}\in H^{1}_{\Delta}(\Omega_{\-/\+}):\gamma^{0,{\-/\+}}_{\Gamma}u_{\-/\+}=0\}\,,$$
be the Dirichlet Laplacian in $L^{2}(\Omega_{\-/\+})$. By \eqref{Green0}, $\Delta^{D}_{\Omega_{\-/\+}}$ is symmetric; in fact it is self-adjoint and the self-adjoint operator $\Delta^{D}_{\Omega_{\-}}\oplus \Delta^{D}_{\Omega_{\+}}$ has an alternative representation:  
 \begin{lemma}\label{Dir} The family of linear bounded maps $\Lambda^{D}$ 
$$ 
\Lambda^{D}_{z}:=P_{z}^{\+}-P_{z}^{\-}:H^{1/2}(\Gamma)\to H^{-1/2}(\Gamma)\,,\qquad z\in \CO\backslash(-\infty,0]\,,
$$ 
satisfies \eqref{Lambda1}-\eqref{Lambda2} and 
$$
\Delta_{\Lambda^{D}}=\Delta^{D}_{\Omega_{\-}}\oplus \Delta^{D}_{\Omega_{\+}}\,.
$$
\end{lemma}
\begin{proof} At first notice that, by the definition \eqref{SL} and by resolvent identity, the operator family $M^{D}_{z}=-\gamma_{\Gamma}^{0}\SL_{z}$, $z\in \CO\backslash(-\infty,0]$, satisfies \eqref{RL}. Then, by \eqref{P-Q-1}, $\Lambda^{D}_{z}=(M_{z}^{D})^{-1}$ and so it satisfies \eqref{Lambda1} and \eqref{Lambda2} by Remark \ref{RemL}.\par 
Let $u\in \dom(\Delta_{\Lambda^{D}})$, so that, by \eqref{domLambda}, $u=u_{z}+G_{z}\Lambda^{D}_{z}\tau=u_{z}+\SL_{z}(P_{z}^{\+}-P_{z}^{\-})\gamma^{0}_{\Gamma}u_{z}$, $u_{z}\in H^{2}(\RE^{n})$. 
By \eqref{SL}, $\SL_{z}\in\B(H^{-1/2}(\Gamma),H^{1}(\RE^{n}))$ and so, since $(-\Delta_{\Omega_{\-/\+}}+z)\SL_{z}=0$, one has $u\in H^{1}(\RE^{n})\cap H^{0}_{\Delta}(\RE^{n}\backslash\Gamma)\subset H^{1}_{\Delta}(\RE^{n}\backslash\Gamma)$. Then, by $H^{1}(\RE^{n})\subset\ker([\gamma^{0}_{\Gamma}])$ and \eqref{P-Q-1}, one gets $\gamma^{0,\-/\+}_{\Gamma}u=0$; therefore $u\in \dom(\Delta^{D}_{\Omega_{\-}})\oplus \dom(\Delta^{D}_{\Omega_{\+}})$ and so $\dom(\Delta_{\Lambda^{D}})\subseteq \dom(\Delta^{D}_{\Omega_{\-}}\oplus\Delta^{D}_{\Omega_{\+}})$. \par
By Theorem \ref{teo}, $\Delta_{\Lambda^{D}}$ is a self-adjoint extension of $(\Delta|\ker(\gamma_{\Gamma}^{0}))\supset \Delta^{\min}_{\Omega_{\-}}\oplus \Delta^{\min}_{\Omega_{\+}}$. Thus $\Delta_{\Lambda^{D}}\subset (\Delta|\ker(\gamma_{\Gamma}^{0}))^{*}\subset (\Delta^{\min}_{\Omega_{\-}}\oplus \Delta^{\min}_{\Omega_{\+}})^{*}=\Delta^{\max}_{\Omega_{\-}}\oplus \Delta^{\max}_{\Omega_{\+}}$. Since $\Delta^{\max}_{\Omega_{\-/\+}}|\dom(\Delta^{D}_{\Omega_{\-/\+}})=\Delta^{D}_{\Omega_{\-/\+}}$, one gets 
$\Delta_{\Lambda^{D}}\subseteq \Delta^{D}_{\Omega_{\-}}\oplus \Delta^{D}_{\Omega_{\+}}$. Since $\Delta_{\Lambda^{D}}$ is self-adjoint and $\Delta^{D}_{\Omega_{\-}}\oplus \Delta^{D}_{\Omega_{\+}}$ is symmetric by \eqref{Green0}, one obtains $\Delta_{\Lambda^{D}}= \Delta^{D}_{\Omega_{\-}}\oplus \Delta^{D}_{\Omega_{\+}}$.
\end{proof}
By Lemma \ref{Dir} and by the compact embedding $H^{-1/2}(\Gamma)\hookrightarrow B^{-3/2}_{2,2}(\Gamma)$, we can apply Theorem \ref{LH}:
\begin{theorem} Let $\Omega$ be a bounded open domain with Lipschitz boundary $\Gamma$. Then asymptotic completeness holds for the scattering couple $(\Delta,\Delta^{D}_{\Omega_{\-}}\oplus \Delta^{D}_{\Omega_{\+}})$ and the corresponding scattering matrix  
$S^{D}_{\lambda}$ is given by 
\begin{equation}\label{scatt-dir}
S^{D}_{\lambda}=\uno-2\pi iL^{D}_{\lambda}(P_{\lambda}^{\+}-P_{\lambda}^{\-})(L^{D}_{\lambda})^{*}\,,\quad\text{$\lambda\in (-\infty,0)\backslash(\sigma_{disc}(\Delta^{D}_{\Omega_{\-}})\cup \sigma_{disc}(\Delta^{D}_{\Omega_{\+}}))$}\,,
\end{equation}
$\sigma_{disc}(\Delta^{D}_{\Omega_{\+}})=\emptyset$ whenever $\Omega_{\+}$ is connected, where 
$$
L^{D}_\lambda: H^{-1/2}(\Gamma)\to L^{2}({\mathbb S}^{n-1})\,,\quad 
L^{D}_{\lambda}\phi(\xi ):=\frac1{2^{\frac12}}\,\frac{|\lambda|^{\frac{n-2}4}}{(2\pi)^{\frac{n}2}}\,\langle u_{\lambda}^{\xi}|\Gamma,\phi \rangle_{H^{1/2}(\Gamma),H^{-1/2}(\Gamma)}\,.
$$
\end{theorem}
\begin{proof} By taking the limit $\epsilon\downarrow 0$ in the identity $-\Lambda^{D}_{\lambda+i\epsilon}\gamma_{\Gamma}^{0}\SL_{\lambda+i\epsilon}=\uno=-
\gamma_{\Gamma}^{0}\SL_{\lambda+i\epsilon}\Lambda^{D}_{\lambda+i\epsilon}
$ and by \eqref{P-Q-1}, one gets $\Lambda^{D,+}_{\lambda}=-(\gamma_{\Gamma}^{0}\SL^{+}_{\lambda})^{-1}=P_{\lambda}^{\+}-P_{\lambda}^{\-}$.\par
Moreover $\sigma_{p}^{-}(\Delta^{D}_{\Omega_{\-}}\oplus \Delta^{D}_{\Omega_{\+}})=
\sigma_{p}(\Delta^{D}_{\Omega_{\-}})\cup \sigma_{p}(\Delta^{D}_{\Omega_{\+}})=
\sigma_{disc}(\Delta^{D}_{\Omega_{\-}})\cup \sigma_{disc}(\Delta^{D}_{\Omega_{\+}})$. Finally, $\sigma_{disc}(\Delta^{D}_{\Omega_{\+}})=\emptyset$ whenever $\Omega_{\+}$ is connected by the unique continuation principle.  
\end{proof}
\begin{remark}
Formula \eqref{scatt-dir} extends to $n$-dimensional bounded Lipschitz domains 
the one which has been obtained, in the case of $2$-dimensional bounded piecewise $C^{2}$ domains, in \cite[Theorems 5.3 and 5.6]{EP1}; similar formulae are also given, in a smooth $2$-dimensional setting in \cite[Subsection 5.2]{BMN} and in a smooth $n$-dimensional setting in \cite[Subsection 6.1]{JST}. Let us mention that as regards the alone asymptotic completeness result, the Lipschitz regularity condition on the boundary is not necessary, see  \cite{Gries}. 
\end{remark}
\end{subsection}
\begin{subsection}{The Laplace operator with Neumann boundary conditions on Lipschitz domains.}\label{neumann} 
Here we apply Theorem \ref{LH} to a case in which $\tau=\gamma^{1}_{\Gamma}$, $\fh=H^{1/2}(\Gamma)$, $\fb=H^{-1/2}(\Gamma)$ and $\Gamma$ is the Lipschitz boundary of a bounded open set $\Omega\subset\RE^{n}$.\par 
Let $$\Delta^{N}_{\Omega_{\-/\+}}:\dom(\Delta^{N}_{\Omega_{\-/\+}})\subset L^{2}(\Omega_{\-/\+})\to L^{2}(\Omega_{\-/\+})\,,\qquad \Delta^{N}_{\Omega_{\-/\+}}u:=\Delta u\,,
$$ 
$$\dom(\Delta^{N}_{\Omega_{\-/\+}}):=\{u_{\-/\+}\in H^{1}_{\Delta}(\Omega_{\-/\+}):\gamma^{1,{\-/\+}}_{\Gamma}u_{\-/\+}=0\}\,,$$
be the Neumann Laplacian in $L^{2}(\Omega_{\-/\+})$. By \eqref{Green1}, $\Delta^{N}_{\Omega_{\-/\+}}$ is symmetric; in fact it is self-adjoint and the self-adjoint operator $\Delta^{N}_{\Omega_{\-}}\oplus \Delta^{N}_{\Omega_{\+}}$ has an alternative representation:
 \begin{lemma}\label{Neu} The family of linear bounded maps $\Lambda^{N}$  
$$ 
\Lambda^{N}_{z}:=Q_{z}^{\-}-Q_{z}^{\+}:H^{-1/2}(\Gamma)\to H^{1/2}(\Gamma)\,,\qquad z\in \CO\backslash(-\infty,0]\,,
$$ 
satisfies \eqref{Lambda1}-\eqref{Lambda2} and 
$$
\Delta_{\Lambda^{N}}=\Delta^{N}_{\Omega_{\-}}\oplus \Delta^{N}_{\Omega_{\+}}\,.
$$
\end{lemma}
\begin{proof} At first notice that, by the definition \eqref{DL} and by resolvent identity, the operator family $M^{N}_{z}=-\gamma_{\Gamma}^{1}\DL_{z}$, $z\in \CO\backslash(-\infty,0]$, satisfies \eqref{RL}. Then $\Lambda^{N}_{z}=(M_{z}^{N})^{-1}$ satisfies \eqref{Lambda1} and \eqref{Lambda2} by Remark \ref{RemL}.\par 
Let $u\in \dom(\Delta_{\Lambda^{N}})$, so that, by \eqref{domLambda}, $u=u_{z}+G_{z}\Lambda^{N}_{z}\tau u_{z}=u_{z}-\DL_{z}(Q_{z}^{\+}-Q_{z}^{\-})\gamma^{1}_{\Gamma}u_{z}$, $u_{z}\in H^{2}(\RE^{n})$. By \cite[Lemma 3.1]{JDE}, $\DL_{z}\in\B(H^{1/2}(\Gamma),H^{1}(\Omega_{\-/\+}))$ and so, since $(-\Delta_{\Omega_{\-/\+}}+z)\DL_{z}=0$, one has $u\in H^{1}_{\Delta}(\RE^{n}\backslash\Gamma)$. Then, by $H^{2}(\RE^{n})\subset(\ker([\gamma^{0}_{\Gamma}])\cap \ker([\gamma^{1}_{\Gamma}]))$, \eqref{jump} and by \eqref{P-Q-2}, one gets $\gamma^{1,\-/\+}_{\Gamma}u=0$; therefore $u\in \dom(\Delta^{N}_{\Omega_{\-}})\oplus \dom(\Delta^{N}_{\Omega_{\+}})$ and so $\dom(\Delta_{\Lambda^{N}})\subseteq \dom(\Delta^{N}_{\Omega_{\-}}\oplus\Delta^{N}_{\Omega_{\+}})$. \par
By Theorem \ref{teo}, $\Delta_{\Lambda^{N}}$ is a self-adjoint extension of $(\Delta|\ker{\gamma_{\Gamma}^{1}})\supset \Delta^{\min}_{\Omega_{\-}}\oplus \Delta^{\min}_{\Omega_{\+}}$. Thus $\Delta_{\Lambda^{N}}\subset (\Delta|\ker{\gamma_{\Gamma}^{1}})^{*}\subset (\Delta^{\min}_{\Omega_{\-}}\oplus \Delta^{\min}_{\Omega_{\+}})^{*}=\Delta^{\max}_{\Omega_{\-}}\oplus \Delta^{\max}_{\Omega_{\+}}$. Since $\Delta^{\max}_{\Omega_{\-/\+}}|\dom(\Delta^{N}_{\Omega_{\-/\+}})=\Delta^{N}_{\Omega_{\-/\+}}$, one gets 
$\Delta_{\Lambda^{N}}\subseteq \Delta^{N}_{\Omega_{\-}}\oplus \Delta^{N}_{\Omega_{\+}}$.  Since $\Delta_{\Lambda^{N}}$ is self-adjoint and $\Delta^{N}_{\Omega_{\-}}\oplus \Delta^{N}_{\Omega_{\+}}$ is symmetric by \eqref{Green1}, one obtains $\Delta_{\Lambda^{N}}= \Delta^{N}_{\Omega_{\-}}\oplus \Delta^{N}_{\Omega_{\+}}$.
\end{proof}
By Lemma \ref{Neu} and by the compact embedding $H^{1/2}(\Gamma)\hookrightarrow H^{-1/2}(\Gamma)$, we can apply Theorem \ref{LH}:
\begin{theorem} Let $\Omega$ be a bounded open domain with Lipschitz boundary $\Gamma$. Then asymptotic completeness holds for the scattering couple $(\Delta,\Delta^{N}_{\Omega_{\-}}\oplus \Delta^{N}_{\Omega_{\+}})$ and the corresponding scattering matrix  
$S^{N}_{\lambda}$ is given by 
\begin{equation}\label{scatt-neu}
S^{N}_{\lambda}=\uno-2\pi iL^{N}_{\lambda}(Q_{\lambda}^{\-}-Q_{\lambda}^{\+})(L^{N}_{\lambda})^{*}\,,\quad\text{$\lambda\in(-\infty,0]\backslash(\sigma_{disc}(\Delta^{N}_{\Omega_{\-}})\cup \sigma_{disc}(\Delta^{N}_{\Omega_{\+}}))$}\,,
\end{equation}
$\sigma_{disc}(\Delta^{N}_{\Omega_{\+}})=\emptyset$ whenever $\Omega_{\+}$ is connected, where 
$$
L^{N}_\lambda: H^{-1/2}(\Gamma)\to L^{2}({\mathbb S}^{n-1})\,,\quad 
L^{N}_{\lambda}\varphi(\xi ):=\frac1{2^{\frac12}}\,\frac{|\lambda|^{\frac{n-2}4}}{(2\pi)^{\frac{n}2}}\,\langle \nu_{\Gamma} \!\cdot\!\nabla u_{\lambda}^{\xi }|\Gamma,\varphi \rangle_{H^{1/2}(\Gamma),H^{-1/2}(\Gamma)}\,.
$$
\end{theorem}
\begin{proof} By taking the limit $\epsilon\downarrow 0$ in the identity $-\Lambda^{N}_{\lambda+i\epsilon}\gamma_{\Gamma}^{1}\DL_{\lambda+i\epsilon}=\uno=-
\gamma_{\Gamma}^{1}\DL_{\lambda+i\epsilon}\Lambda^{N}_{\lambda+i\epsilon}
$ and by \eqref{P-Q-1}, one gets $\Lambda^{N,+}_{\lambda}=-(\gamma_{\Gamma}^{1}\DL^{+}_{\lambda})^{-1}=Q_{\lambda}^{\-}-Q_{\lambda}^{\+}$.\par
Moreover $\sigma_{p}^{-}(\Delta^{N}_{\Omega_{\-}}\oplus \Delta^{N}_{\Omega_{\+}})\cup\{0\}=
\sigma_{p}(\Delta^{N}_{\Omega_{\-}})\cup \sigma_{p}(\Delta^{N}_{\Omega_{\+}})=
\sigma_{disc}(\Delta^{N}_{\Omega_{\-}})\cup \sigma_{disc}(\Delta^{N}_{\Omega_{\+}})$. Finally, $\sigma_{disc}(\Delta^{N}_{\Omega_{\+}})=\emptyset$ whenever $\Omega_{\+}$ is connected by the unique continuation principle.  

\end{proof}
\begin{remark}
Formula \eqref{scatt-neu} extends to $n$-dimensional bounded Lipschitz domains 
the one which has been obtained, in the case of $2$-dimensional bounded piecewise $C^{2}$ domains, in \cite[Theorems 4.2 and 4.3]{EP3}; similar formulae are also given, in a smooth $2$-dimensional setting in \cite[Subsection 5.3]{BMN} and in a smooth $n$-dimensional setting in \cite[Subsection 6.2]{JST}
\end{remark}
\end{subsection} 
\begin{subsection} {The Laplace operator with semi-transparent boundary conditions of $\delta$-type  on $d$-sets.} Here we apply Theorem \ref{LH} to a case in which $\tau=\gamma^{0}_{\Gamma}$, $\fh=B^{s_{d}}_{2,2}(\Gamma)$, $s_{d}:={2-(n-d)/2}$, $\fb=H^{s}(\Gamma)$, $0<s<s_{d}-1$, and $\Gamma\subset\RE^{n}$ is a $d$-set with $0<n-d<2$. 
\begin{lemma}\label{lambdalpha} Let $\alpha\in\B(H^{s}(\Gamma),H^{-s}(\Gamma))$, $\alpha^{*}=\alpha$, $0<s<1-\frac{n-d}2$. Then there exists a finite set $\Sigma_{\alpha} \subset(0,+\infty)$ such that for all $z\in\CO\backslash((-\infty,0]\cup \Sigma_{\alpha} )$ one has $(\uno+\alpha \gamma_{\Gamma}^{0}\SL_{z})^{-1}\in \B(H^{-s}(\Gamma))$. Moreover the operator family $\Lambda^{\alpha}$ in $\B(H^{s}(\Gamma),H^{-s}(\Gamma)))$ given by $$\Lambda^{\alpha}_{z}:=-(\uno+\alpha\gamma_{\Gamma}^{0}\SL_{z})^{-1}\alpha\,,\quad z\in\CO\backslash((-\infty,0]\cup \Sigma_{\alpha} )\,,
$$
satisfies \eqref{Lambda1} and \eqref{Lambda2}.
\end{lemma}
\begin{proof} By Fourier transform, one has the following estimate holding for any $z\in\CO\backslash(-\infty,0]$ and for any real number $s$:
$$
{\|}R_{z}^{0}{\|}_{H^{s}(\RE^{n}),H^{s+t}(\RE^{n})}\le \frac1{{d}_{z}^{1-\frac{t}2}}\,,\qquad  0\le t\le 2\,,
$$
where $d_{z}:=\text{dist}(z,(-\infty,0])$. Thus, by the mapping properties of $\gamma_{\Gamma}^{0}$ and $(\gamma_{\Gamma}^{0})^{*}$,  one gets 
\begin{align*}
&\ {\|}\gamma_{\Gamma}^{0}R_{z}^{0}(\gamma_{\Gamma}^{0})^{*}\ {\|}_{B_{2,2}^{-s}(\Gamma),B_{2,2}^{-s-(n-d)+t}(\Gamma)}\\
\le& \frac1{d_{z}^{1-\frac{t}2}}\ {\|}(\gamma_{\Gamma}^{0})^{*}{\|}_{B_{2,2}^{-s}(\Gamma),B_{2,2}^{-s-\frac{n-d}2}(\Gamma)}\ {\|} \gamma_{\Gamma}^{0}{\|}_{B_{2,2}^{-s-\frac{n-d}2+t}(\Gamma),B_{2,2}^{-s-(n-d)+t}(\Gamma)}\,.
\end{align*}
Choosing $t=2s+n-d$, such an inequality shows that if $0<s+\frac{n-d}2<1$ then there exists $c_{\alpha}>0$ such that  operator norm ${\|}\gamma_{\Gamma}^{0}\SL_{z}\alpha{\|}_{H^{s}(\Gamma),H^{s}(\Gamma)}$ is strictly smaller than one whenever $\text{\rm Re}(z)>c_{\alpha}$. 
Therefore $(\uno+\gamma_{\Gamma}^{0}\SL_{z}\alpha)^{-1}\in \B(H^{s}(\Gamma))$ whenever $\text{\rm Re}(z)>c_{\alpha}$. 
\par
Let $0<s+\frac{n-d}2<1$. By \eqref{emb}, the embedding $B_{2,2}^{2-s-(n-d)}(\Gamma)\hookrightarrow B_{2,2}^{s}(\Gamma)$  is compact and so, by $\ran(\gamma_{\Gamma}^{0}\SL_{z})\subseteq B_{2,2}^{2-s-(n-d)}(\Gamma)$, the map 
$\gamma_{\Gamma}^{0}\SL_{z}:H^{-s}(\Gamma)\to H^{s}(\Gamma)$ is also compact; thus   $\gamma_{\Gamma}^{0}\SL_{z}\alpha:H^{s}(\Gamma)\to H^{s}(\Gamma)$ is compact as well. 
Since the map $z\mapsto  \gamma_{\Gamma}^{0}\SL_{z}\alpha$ is analytic from $z\in\CO\backslash(-\infty,0]$ to $\B(H^{s}(\Gamma))$ and the set of $z\in\CO\backslash((-\infty,0]$ such that $(\uno+ \gamma_{\Gamma}^{0}\SL_{z}\alpha)^{-1}\in \B(H^{s}(\Gamma))$ is not void, by analytic Fredholm theory (see e.g. \cite[Theorem XIII.13]{RS-IV}), $(\uno+ \gamma_{\Gamma}^{0}\SL_{z}\alpha)^{-1}\in \B(H^{s}(\Gamma))$ for any $z\in\CO\backslash((-\infty,0]\cup \Sigma_{\alpha} )$, where $\Sigma_{\alpha} $ is a discrete set.  By Theorem \ref{teo}, Remark \ref{rem} and next Theorem \ref{delta} (see \eqref{krein-dset}), $\Sigma_{\alpha} $ is contained in the spectrum of a self-adjoint operator and so $\Sigma_{\alpha} \subset\RE$; hence $\Sigma_{\alpha} \subseteq [0,c_{\alpha}]$ and so it is finite being discrete, i.e. without accumulation points.
\par
By $\alpha=\alpha^{*}$ and the same arguments as in the proof of \cite[Corollary 2.4]{Acu}, one obtains \eqref{Lambda1}
and 
\be\label{WF}
(\uno+\alpha\gamma_{\Gamma}^{0}\SL_{z})^{-1}=\big((\uno+\gamma_{\Gamma}^{0}\SL_{\bar z}\alpha)^{-1}\big)^{*}\in\B(H^{-s}(\Gamma))\,.
\ee 
Finally, by $\SL_{z}=R_{z}^{0}(\gamma_{\Gamma}^{0})^{*}$ and resolvent identity for $R^{0}_{z}$,  it results%
\begin{equation*}
(  1+\alpha\gamma_{\Gamma}^{0}\SL_{w})  -(  1+\alpha\gamma_{\Gamma}^{0}
\SL_{z})  =(  z-w)  \alpha\gamma_{\Gamma}^{0}R_w^{0}\SL_{z}\,.
\end{equation*}
This yields%
\begin{equation*}
\Lambda^{\alpha}_{w}  -\Lambda^{\alpha}_{z} =(  z-w)
\Lambda^{\alpha}_{w}\gamma_{\Gamma}^{0}R_{w}^{0}\SL_{z}\Lambda^{\alpha}_{z}
\end{equation*}
i.e. relation \eqref{Lambda2}.
\end{proof}
\vskip8pt
Taking $\lambda_{\circ}>0$, in the following we use the shorthand notation $\SL_{\circ}\equiv \SL_{\lambda_{\circ}}$.
\begin{theorem}\label{delta} Let $\Gamma$ be a $d$-set with $n-2<d<n$ and let $\alpha\in\B(H^{s}(\Gamma),H^{-s}(\Gamma))$, $\alpha^{*}=\alpha$, $0<s<1-\frac{n-d}2$. Then\par\noindent
1) The family of bounded linear operators 
\be\label{krein-dset}
R_{z}^{\alpha}:=R_{z}^{0}-\SL_{z}(\uno+\alpha\gamma_{\Gamma}^{0}\SL_{z})^{-1}\alpha\gamma_{\Gamma}^{0}R_{z}^{0}\,,\quad z\in\CO\backslash((-\infty,0]\cup \Sigma_{\alpha} )
\ee
is the resolvent of the bounded from above self-adjoint operator $\Delta_{\alpha}$ in $L^{2}(\RE^{n})$ defined, in a $\lambda_{\circ}$-independent way, by
\be\label{dom-alfa}
\dom(  \Delta_{\alpha})  :=\{ u\in H^{2-s-\frac{n-d}2}(\RE^{n}): u+\SL_{{\circ}}\alpha\gamma_{\Gamma}^{0}u\in
H^{2}(  \mathbb{R}^{n})\}  \,,
\ee
\be\label{Aalfa}
\Delta_{\alpha}u:=\Delta u-(\gamma_{\Gamma}^{0})^{*}\alpha\gamma_{\Gamma}^{0}u\,.
\ee
2) $\sigma_{ess}(\Delta_{\alpha})=\sigma_{ac}(\Delta_{\alpha})=(-\infty,0]$, $\sigma_{disc}(\Delta_{\alpha})=\Sigma_{\alpha} $ is finite, $\sigma_{sc}(\Delta_{\alpha})=\emptyset$, $\sigma_{p}^{-}(\Delta_{\alpha}):=(-\infty,0)\cap \sigma_{p}(\Delta_{\alpha})$ is at most discrete and  asymptotic completeness holds for the scattering couple $(\Delta,\Delta_{\alpha})$. \par\noindent 
3) The inverse $(\uno_{\ran(\alpha)}+\alpha \gamma_{\Gamma}^{0}\SL^{\pm}_{\lambda})^{-1}:\ran(\alpha)\to \ran(\alpha)$ exists for any $\lambda\in(-\infty,0)\backslash\sigma^{-}_{p}(\Delta_{\alpha})$  
and the  scattering matrix $S^{\alpha}_{\lambda}$ is given by 
\begin{equation}\label{scatt-delta}
S^{\alpha}_{\lambda}=\uno+2\pi iL^{D}_{\lambda}(\uno_{\ran(\alpha)}+\alpha \gamma_{\Gamma}^{0}\SL^{+}_{\lambda})^{-1} \alpha (L^{D}_{\lambda})^{*}\,,\quad\text{ $\lambda\in(-\infty,0)\backslash\sigma^{-}_{p}(\Delta_{\alpha})$}\,,
\end{equation}
$$
L^{D}_\lambda: H^{-s}(\Gamma)\to L^{2}({\mathbb S}^{n-1})\,,\quad 
L^{D}_{\lambda}\phi(\xi ):=\frac1{2^{\frac12}}\,\frac{|\lambda|^{\frac{n-2}4}}{(2\pi)^{\frac{n}2}}\,\langle u_{\lambda}^{\xi}|\Gamma,\phi \rangle_{H^{s}(\Gamma),H^{-s}(\Gamma)}\,.
$$
\end{theorem}
\begin{proof} By Lemma \ref{lambdalpha}, we can apply Theorem \ref{teo} and $\Delta_{\alpha}:=\Delta_{\Lambda^{\alpha}}$ is a well defined self-adjoint operator with resolvent given by \eqref{resolvent}. By \eqref{resolvent} and Lemma \ref{lambdalpha}, one gets $\sigma(\Delta_{\alpha})\subseteq(-\infty,\sup\Sigma_{\alpha}]$ and so $\Delta_{\alpha}$ is bounded from above since $\Sigma_{\alpha}$ is finite. By Lemma \ref{lambdalpha}, $\ran(\Lambda_{z}^{\alpha})\hookrightarrow H^{-s}(\Gamma)\hookrightarrow B_{2,2}^{-2+\frac{n-d}2}(\Gamma)$ and so $\ran(\Lambda_{z}^{\alpha})$ is compactly embedded in $\fh^{*}=B_{2,2}^{-2+\frac{n-d}2}(\Gamma)$ by \eqref{emb}. Since $\Gamma$ is bounded, \eqref{HT} hold true. Therefore hypotheses i)-iii) in Theorem \ref{LH} hold.\par 
By \eqref{krein-dset} and \cite[Theorem XIII.13]{RS-IV}, $z\mapsto R^{\alpha}_{z}$ has poles (and the coefficients of the Laurent expansion are finite-rank operators) only at $\Sigma_{\alpha} $; so,  by \cite[Lemma 1, page 108]{RS-IV}, $\sigma_{disc}(\Delta_\alpha)=\Sigma_{\alpha} $.
\par The proofs of  
\eqref{dom-alfa} and \eqref{Aalfa} are the same as the ones given (in the case $\Gamma$ is Lipschitz) in the proof of \cite[Theorem 2.5]{Acu} and are not reproduced here.\par
Considering the limit $\epsilon\downarrow 0$ in the identity $\Lambda^{\alpha}_{\lambda\pm i\epsilon}=-(\uno+\alpha\gamma_{\Gamma}^{0}\SL_{\lambda\pm i\epsilon})^{-1}\alpha$, one gets $\ker(\alpha)\subseteq\ker(\Lambda^{\alpha,\pm}_{\lambda})$. Considering the limit $\epsilon\downarrow 0$ in the identity 
$$
-(\uno+\alpha\gamma_{\Gamma}^{0}\SL_{\lambda\pm i\epsilon})\Lambda^{\alpha}_{\lambda\pm i\epsilon}=\alpha=-\Lambda^{\alpha}_{\lambda\pm i\epsilon}(\uno+\gamma_{\Gamma}^{0}\SL_{\lambda\pm i\epsilon}\alpha)\,,$$ 
one gets
$$
-(\uno+\alpha\gamma_{\Gamma}^{0}\SL^{\pm}_{\lambda})\Lambda^{\alpha,\pm}_{\lambda}=\alpha=-\Lambda^{\alpha,\pm}_{\lambda}(\uno+\gamma_{\Gamma}^{0}\SL^{\pm}_{\lambda}\alpha)\,,$$ 
and
 $$-(\tilde\alpha^{-1}+\gamma_{\Gamma}^{0}\SL^{\pm}_{\lambda})\Lambda^{\alpha,\pm}_{\lambda}|\ker(\alpha)^{\perp}=\uno_{\ker(\alpha)^{\perp}}\,,\qquad -\Lambda^{\alpha,\pm}_{\lambda}(\tilde\alpha^{-1}+\gamma_{\Gamma}^{0}\SL^{\pm}_{\lambda})|\ran(\alpha)=\uno_{\ran(\alpha)}\,,
$$ 
where $\tilde\alpha:\ker(\alpha)^{\perp}\to\ran(\alpha)$ is the bijective bounded linear operator $\tilde\alpha:=\alpha|\ker(\alpha)^{\perp}$.
This shows that $\ran(\Lambda^{\alpha,\pm}_{\lambda})\subseteq\ran(\alpha)$ and that  $\tilde\alpha^{-1}+\gamma_{\Gamma}^{0}\SL^{\pm}_{\lambda}:\ran(\alpha)\to \ker(\alpha)^{\perp}$ is invertible with inverse $-\Lambda^{\alpha,\pm}_{\lambda}|\ker(\alpha)^{\perp}$, i.e. 
$$
\Lambda^{\alpha,\pm}_{\lambda}|\ker(\alpha)^{\perp}=-(\tilde\alpha^{-1}+\gamma_{\Gamma}^{0}\SL^{\pm}_{\lambda})^{-1}:\ker(\alpha)^{\perp}\to\ran(\alpha)\,. 
$$
By $(\tilde\alpha^{-1}+\gamma_{\Gamma}^{0}\SL^{\pm}_{\lambda})^{-1}\tilde\alpha^{-1}=
(\alpha(\tilde\alpha^{-1}+\gamma_{\Gamma}^{0}\SL^{\pm}_{\lambda}))^{-1}=
(\uno_{\ran(\alpha)}+\alpha\gamma_{\Gamma}^{0}\SL^{\pm}_{\lambda})^{-1}$
one gets the existence of the inverse 
$$
(\uno_{\ran(\alpha)}+\alpha \gamma_{\Gamma}^{0}\SL^{\pm}_{\lambda})^{-1}:\ran(\alpha)\to \ran(\alpha)
$$
and the identity
$$
\Lambda_{\lambda}^{\alpha,\pm}=-(\uno_{\ran(\alpha)}+\alpha \gamma_{\Gamma}^{0}\SL^{\pm}_{\lambda})^{-1}\alpha:H^{s}(\Gamma)\to H^{-s}(\Gamma)\,.
$$
\end{proof}
\begin{remark} The limit single-layer operator $\SL_{\lambda}^{\pm}$ admits the representation  
\begin{equation*}
\SL^{\pm}_{\lambda}\phi(x)=\frac{i}{4}\,\int_{\Gamma}\left(\frac{\mp|\lambda|^{1/2}}{2\pi\|x-y\|}\,\right)^{\frac{n}2-1}\!\!\!\!H^{(1)}_{\frac{n}2-1}(\mp|\lambda|^{1/2}\,\|x-y\|)
\,\phi(y)\,d\mu^{d}_{\Gamma} (y)
\end{equation*}
whenever $\phi\in L^{2}(\Gamma)$ and $x\notin\Gamma$, where  $H^{(1)}_{\frac{n}2-1}$ denotes the Hankel function of first kind of order $\frac{n}2-1$ (see \cite[equation (5.1)]{JST}).
\end{remark}
\begin{remark}\label{multi} A particular case of operator $\alpha\in \B((H^{s}(\Gamma),H^{-s}(\Gamma))$, such that $\alpha=\alpha^{*}$ is $\alpha\in M(H^{s}(\Gamma),H^{-s}(\Gamma)) $, $\alpha$ real-valued, where $M(H^{s_{1}}(\Gamma),H^{s_{2}}(\Gamma))$ denotes the set of Sobolev multipliers on  $H^{s_{1}}(\Gamma)$ to $H^{s_{2}}(\Gamma)$ (here and in the following we use the same notation for a function and for the corresponding multiplication operator). By the inequality
$$
\left|\int_{\Gamma}\alpha\bar \phi\psi\,d\mu_{\Gamma}^{d}\right|\le \||\alpha|^{1/2}\phi\|_{L^{2}(\Gamma)}\||\alpha|^{1/2}\psi\|_{L^{2}(\Gamma)}\le \n|\alpha|^{1/2}\n^{2}_{H^{s}(\Gamma),L^{2}(\Gamma)}\|\phi\|_{H^{s}(\Gamma)}\|\psi\|_{H^{s}(\Gamma)}\,,
$$
one has
$$
{|\alpha|^{1/2}}\in M(H^{s}(\Gamma),L^{2}(\Gamma))\quad\Longrightarrow\quad\alpha\in M(H^{s}(\Gamma),H^{-s}(\Gamma))\,.
$$
Then, by the embeddings \eqref{embLp} and H\"older's inequality, one gets
$$p\ge \frac{1}{s} \quad\Longrightarrow\quad L^{p}(\Gamma)\subseteq M(H^{s}(\Gamma),H^{-s}(\Gamma)) \,.
$$
Thus we can define $\Delta_{\alpha}$ for any real-valued $\alpha\in L^{p}(\Gamma)$, $p>\frac2{2-(n-d)}$.
\end{remark}
In the case $\Gamma$ in Theorem \ref{scatt-delta} is a $(n-1)$-set which is the boundary of a bounded Lipschitz domain, some of the results in the previous theorem can be improved:
\begin{corollary}\label{deltaLip} Let $\Omega$ be an open bounded set with a Lipschitz boundary 
$\Gamma$ and $\alpha\in \B(H^{s}(\Gamma),H^{-s}(\Gamma))$, $\alpha=\alpha^{*}$, $0<s<1/2$. Then\par\noindent
\be\label{delta_res}
(-\Delta_{\alpha}+z)^{-1}=R_{z}^{0}+\SL_{z}(P_{z}^{\+}-P_{z}^{\-})(\alpha-(P_{z}^{\+}-P_{z}^{\-}))^{-1}\alpha\gamma_{\Gamma}^{0}R_{z}^{0}\,,\quad z\in\CO\backslash\big((-\infty,0]\cup \Sigma_{\alpha}\big)\,,
\ee
\be
\dom(  \Delta_{\alpha})=\{ u\in H^{3/2-s}(\RE^{n})\cap H_{\Delta}^{0}(\RE^{n}\backslash\Gamma): \alpha\gamma_{\Gamma}^{0}u=[\gamma_{\Gamma}^{1}]u\}\,.
\ee
$$
\Delta_{\alpha}u=(\Delta_{\Omega_{\-}}\oplus\Delta_{\Omega_{\+}})u\,.
$$
Whenever $\lambda\in(-\infty,0)\backslash(\sigma^{-}_{p}(\Delta_{\alpha})\cup\sigma_{disc}(\Delta^{D}_{\Omega_{\-}})\cup\sigma_{disc}(\Delta^{D}_{\Omega_{\+}}))$, the  scattering matrix $S^{\alpha}_{\lambda}$ has the alternative representation  
\begin{equation}\label{scatt-delta-lip}
S^{\alpha}_{\lambda}=\uno-2\pi iL^{D}_{\lambda}
(P_{\lambda}^{\+}-P_{\lambda}^{\-})(\alpha-(P_{\lambda}^{\+}-P_{\lambda}^{\-}))^{-1}\alpha\, (L^{D}_{\lambda})^{*}\,.
\end{equation}
If $\Omega_{\+}$ is connected then $\sigma_{p}^{-}(\Delta_{\alpha})=\sigma_{disc}(\Delta^{D}_{\Omega_{\+}})=\emptyset$. 
\end{corollary}
\begin{proof} Relation \eqref{delta_res} is consequence of \eqref{P-Q-1}: by
$$
(\alpha-\Lambda^{D}_{z})\gamma_{\Gamma}^{0}\SL_{z}(\uno+\alpha\gamma_{\Gamma}^{0}\SL_{z})^{-1}=\uno=\gamma_{\Gamma}^{0}\SL_{z}(\uno+\alpha\gamma_{\Gamma}^{0}\SL_{z})^{-1}(\alpha-\Lambda^{D}_{z})\,,
$$
one gets 
$(\alpha-\Lambda^{D}_{z})^{-1}=\gamma_{\Gamma}^{0}\SL_{z}(\uno+\alpha\gamma_{\Gamma}^{0}\SL_{z})^{-1}$ and so 
$$
\Lambda_{z}^{\alpha}=\Lambda^{D}_{z}(\alpha-\Lambda^{D}_{z})^{-1}\alpha=
(P_{z}^{\+}-P_{z}^{\-})(\alpha-(P_{z}^{\+}-P_{z}^{\-}))^{-1}\alpha\,.$$
By $H^{2}(\RE^{n})\subseteq\ker([\gamma_{\Gamma}^{1}])$ and by \eqref{jump}, one gets $\dom(  \Delta_{\alpha})  \subseteq D_{\alpha}$, where $$D_{\alpha}:=\{ \psi\in H^{3/2-s}(\RE^{n})\cap H^{0}_{\Delta}(\RE^{n}\backslash\Gamma): \alpha\gamma_{\Gamma}^{0}u=[\gamma_{\Gamma}^{1}]u\}\,.
$$ 
Thus $\Delta_{\alpha}\subseteq (\Delta_{\Omega_{\-}}\oplus\Delta_{\Omega_{\+}})|D_{\alpha}$. 
Since $\Delta_{\alpha}$ is self-adjoint and $(\Delta_{\Omega_{\-}}\oplus\Delta_{\Omega_{\+}})|D_{\alpha}$ is symmetric by \eqref{Green0}, the  two operators coincide.
\par
By the same reasonings as in the proof of  Theorem \ref{delta}, one shows that $(\alpha+\Lambda_{\lambda}^{D,\pm})|\ran(\alpha)$ is invertible and that 
$$
\Lambda_{\lambda}^{\alpha,\pm}=\Lambda^{D,\pm}_{\lambda}(\alpha-\Lambda^{D,\pm}_{\lambda})^{-1}\alpha=
(P_{\lambda}^{\+}-P_{\lambda}^{\-})(\alpha-(P_{\lambda}^{\+}-P_{\lambda}^{\-}))^{-1}\alpha\,.$$
If $\Omega_{\+}$ is connected, then $\sigma_{p}(\Delta_{\alpha})\cap (-\infty,0)=\emptyset$ by the unique continuation principle (see \cite[Remark 3.8]{JST}). 
\end{proof}
\begin{remark}
The conditions providing the self-adjoint operator $\Delta_{\alpha}$ in Theorem \ref{delta} are weaker, as regards the regularity of the boundary and/or the class of admissible strength functions, than the ones assumed in previous works, see, for example, \cite{Her}, \cite{F92}, \cite{BEKS}, \cite{P01}, \cite{BLL}, \cite{JDE}, \cite{ER}. Asymptotic completeness for the scattering couple $(\Delta,\Delta_{\alpha})$ provided in Theorem \ref{delta} extend results on existence and completeness given, in the case the boundary is smooth and the strength are bounded, in \cite{BLL} and \cite{JDE}. The formula for the scattering matrix provided in \eqref{scatt-delta} (respectively in \eqref{scatt-delta-lip}) extends to $d$-sets (respectively to Lipschitz hypersurfaces) the results given, in the case of a smooth hypersurface, in \cite[Subsections 6.4 and 7.4]{JST} and, in the case of a smooth $2$- or   $3$-dimensional hypersurface, in \cite[Subsection 5.4]{BMN} (see also the formula provided in \cite{F93} for Schr\"odinger operators of the kind $-\Delta+\mu$, $\mu$ a signed measure).
\end{remark}
\end{subsection}
\begin{subsection} {The Laplace operator with semi-transparent boundary condition of $\delta'$-type on Lipschitz hypersurfaces.} Here we apply Theorem \ref{LH} to a case in which $\tau=\gamma^{1}_{\Gamma}$, $\fh=H^{1/2}(\Gamma)$, $\fb=H^{-1/2}(\Gamma)$ and $\Gamma$ is the boundary of a bounded Lipschitz set $\Omega\subset\RE^{n}$. 
\begin{lemma}\label{lambdateta} Let $\theta\in\B(H^{s}(\Gamma),H^{-s}(\Gamma))$, $\theta^{*}=\theta$, $0<s<\frac12$. Then there exists a discrete set $\Sigma_{\theta} \subset(0,+\infty)$ such that for all $z\in\CO\backslash((-\infty,0]\cup \Sigma_{\theta} )$ one has $(\uno+\theta(Q^{\-}_{z}-Q^{\+}_{z}))^{-1}\in \B(H^{-1/2}(\Gamma))$. Moreover the operator family $\Lambda^{\theta}$ in $\B(H^{-1/2}(\Gamma),H^{1/2}(\Gamma)))$ given by $$\Lambda^{\theta}_{z}:=(Q^{\-}_{z}-Q^{\+}_{z})
(\uno+\theta(Q^{\-}_{z}-Q^{\+}_{z}))^{-1}
\,,\quad z\in\CO\backslash((-\infty,0]\cup \Sigma_{\theta} )\,,
$$
satisfies \eqref{Lambda1} and \eqref{Lambda2}.
\end{lemma}
\begin{proof} 
By $\theta({Q}_{z}^{\-}-{Q}_{z}^{\+})  \in{\B}(  H^{-1/2}(  \Gamma)  ,
H^{-s}(\Gamma) )  $ and by the compact embedding $H^{-s}(
\Gamma)  \hookrightarrow H^{-1/2}(  \Gamma)  $, one gets 
$\theta({Q}_{z}^{\-}-{Q}_{z}^{\+})\in{\mathfrak S}_{\infty}(  H^{-1/2}(  \Gamma))  $. Therefore, by the Fredholm
alternative, $ \uno+\theta( {Q}_{z}^{\-}-{Q}_{z}^{\+})  $ is invertible if and only if  its kernel $K_{z}$ is trivial.
By ${Q}_{z}^{\+}-{Q}_{z}^{\-}=(\gamma_{\Gamma}^{1}\DL_{z})^{-1}$, $K_{z}\not=\{0\}$ if and only if there is $\psi\in H^{1/2}(\Gamma)\backslash\{0\}$ such that 
\be\label{Kz}
\theta\psi=\gamma_{\Gamma}^{1}\DL_{z}\psi\,.
\ee
By the definition \eqref{DL} and by resolvent identity, we have
\be\label{gG}
\gamma_{\Gamma}^{1}\DL_{z}-(\gamma_{\Gamma}^{1}\DL_{z})^{*}=\gamma_{\Gamma}^{1}\DL_{z}-\gamma_{\Gamma}^{1}\DL_{\bar z}=(\bar z-z)\DL_{z}^{*}\DL_{z}\,.\ee 
Since $\theta=\theta^{*}$, \eqref{Kz} and \eqref{gG} entail, for any $z\in \CO\backslash\RE$, $$0=(z-\bar z)\|\DL_{z}\psi\|^{2}_{L^{2}(\RE^{n})}\,.$$ 
Since $\DL_{z}^{*}=\gamma_{\Gamma}^{1}R^{0}_{z}$ is surjective, $\DL_{z}$ has closed range by the closed range theorem and so (see Remark \ref{RMM}) there exists $c>0$ such that $\|\DL_{z}\psi\|_{L^{2}(\RE^{n})}\ge c\, \|\psi\|_{H^{1/2}(\Gamma)}$. Thus $K_{z}=\{0\}$ whenever $z\in \CO\backslash\RE$ and $\uno+\theta( {Q}_{z}^{\-}-{Q}_{z}^{\+})$ has a bounded inverse for any $z\in \CO\backslash\RE$. Since the operator-valued map $z\mapsto \uno+\theta( {Q}_{z}^{\-}-{Q}_{z}^{\+})$ is analytic on $\CO\backslash(-\infty,0]$, by analytic Fredholm theory (see e.g. \cite[Theorem XIII.13]{RS-IV}), $\uno+\theta( {Q}_{z}^{\-}-{Q}_{z}^{\+})^{-1}\in \B(H^{-1/2}(\Gamma))$ for any $z\in\CO\backslash((-\infty,0]\cup \Sigma_{\theta} )$, where $\Sigma_{\theta} \subset(0,+\infty)$ is a discrete set.\par
Since 
\be\label{gt0}
\Lambda^{\theta}_{z}=\Lambda^{N}_{z}(\uno+\theta\Lambda^{N}_{z})^{-1}\,,
\ee
one has
\be\label{gt}
(\Lambda^{\theta}_{z})^{-1}=(\uno+\theta\Lambda^{N}_{z})(\Lambda^{N}_{z})^{-1}
=\theta+(\Lambda^{N}_{z})^{-1}
\ee
Thus $\Lambda^{\theta}_{z}$ satisfies \eqref{Lambda1} and \eqref{Lambda2} by Remark \ref{RemL} and Lemma \ref{Neu}.
\end{proof}
\vskip8pt
Taking $\lambda_{\circ}>0$, in the following we use the shorthand notation $\DL_{\circ}\equiv \DL_{\lambda_{\circ}}$.
\begin{theorem}\label{delta'} Let $\Gamma$ be the boundary of a bounded Lipschitz set $\Omega\subset\RE^{n}$ and let $\theta\in\B(H^{s}(\Gamma),H^{-s}(\Gamma))$, $\theta^{*}=\theta$, $0<s<\frac12$. Then\par\noindent
1) The family of bounded linear operators 
\be\label{krein}
R_{z}^{\theta}:=R_{z}^{0}+\DL_{z}(Q^{\-}_{z}-Q^{\+}_{z})
(\uno+\theta(Q^{\-}_{z}-Q^{\+}_{z}))^{-1}\gamma_{\Gamma}^{1}R_{z}^{0}\,,\quad z\in\CO\backslash((-\infty,0]\cup \Sigma_{\theta} )
\ee
is the resolvent of the bounded from above self-adjoint operator $\Delta_{\theta}$ in $L^{2}(\RE^{n})$ given by $$
\dom(  \Delta_{\theta})=\{ u\in H^{1}_{\Delta}(\RE^{n}): [\gamma_{\Gamma}^{1}]u=0\,,\ \gamma_{\Gamma}^{1}u=\theta[\gamma_{\Gamma}^{0}]u\}\,,
$$
$$
\Delta_{\theta}u=(\Delta_{\Omega_{\-}}\oplus\Delta_{\Omega_{\+}})u\,.
$$
2) $\sigma_{ess}(\Delta_{\theta})=\sigma_{ac}(\Delta_{\theta})=(-\infty,0]$, $\sigma_{disc}(\Delta_{\theta})=\Sigma_{\theta} $ is finite, $\sigma_{sc}(\Delta_{\theta})=\emptyset$, $\sigma^{-}_{p}(\Delta_{\theta})=(-\infty,0)\cap\sigma_{p}(\Delta_{\theta})$ is at most discrete.
\par\noindent 3) Asymptotic completeness holds for the scattering couple $(\Delta,\Delta_{\theta})$ and, whenever  $\lambda\in(-\infty,0)\backslash(\sigma^{-}_{p}(\Delta_{\theta})\cup\sigma_{disc}(\Delta^{N}_{\Omega_{\-}})\cup\sigma_{disc}(\Delta^{N}_{\Omega_{\+}}))$, the scattering matrix $S^{\theta}_{\lambda}$ is given by 
\begin{equation}\label{scatt-delta'}
S^{\theta}_{\lambda}=\uno-2\pi iL^{N}_{\lambda}(Q^{\-}_{\lambda}-Q^{\+}_{\lambda})
(\uno+\theta(Q^{\-}_{\lambda}-Q^{\+}_{\lambda}))^{-1}(L^{N}_{\lambda})^{*}\,.
\end{equation}
If $\Omega_{\+}$ is connected, then $\sigma^{-}_{p}(\Delta_{\theta})=\sigma_{disc}(\Delta^{N}_{\Omega_{\+}})=\emptyset$.
\end{theorem}
\begin{proof} 
By \eqref{domLambda}, $u$ belongs to $\dom(  \Delta_{\theta})$ if and only if $u=u_{z}+\DL_{z}\Lambda^{\theta}_{z}\gamma^{1}_{\Gamma}u_{z}$. By \cite[Lemma 3.1]{JDE}, $\DL_{z}\in\B(H^{1/2}(\Gamma),H^{1}(\Omega_{\-/\+}))$ and so, since $(-(\Delta_{\Omega_{\-}}\oplus\Delta_{\Omega_{\-}})+z)\DL_{z}=0$, one has $u\in H^{1}_{\Delta}(\RE^{n}\backslash\Gamma)$. Then, by $H^{2}(\RE^{n})\subset(\ker([\gamma^{0}_{\Gamma}])\cap\ker([\gamma^{1}_{\Gamma}]))$ and \eqref{jump}, one obtains $[\gamma^{0}_{\Gamma}]u=\Lambda_{z}^{\theta}\gamma_{\Gamma}^{1}u_{z}$ and $[\gamma^{1}_{\Gamma}]u=0$. Moreover, by \eqref{gt}, 
$(\Lambda^{\theta}_{z})^{-1}=\theta-\gamma^{1}_{\Gamma}\DL_{z}$, and so $(\theta-\gamma^{1}_{\Gamma}\DL_{z})[\gamma_{\Gamma}^{0}]u=\gamma_{\Gamma}^{1}u_{z}$; thus $\gamma_{\Gamma}^{1}u=\theta[\gamma_{\Gamma}^{0}]u$ and 
$$
\dom(  \Delta_{\theta})  \subseteq D_{\theta}:=\{ u\in H^{1}_{\Delta}(\RE^{n}\backslash\Gamma):[\gamma^{1}_{\Gamma}]u=0,\ \gamma_{\Gamma}^{1}u=\theta[\gamma_{\Gamma}^{0}]u\}. 
$$
Therefore $\Delta_{\theta}\subseteq (\Delta_{\Omega_{\-}}\oplus\Delta_{\Omega_{\+}})|D_{\theta}$. 
Since $\Delta_{\theta}$ is self-adjoint and $(\Delta_{\Omega_{\-}}\oplus\Delta_{\Omega_{\+}})|D_{\theta}$ is symmetric by \eqref{Green1}, the  two operators coincide.
\par
By Green's formula (see e.g. \cite[Theorem 4.4]{McL}) and by Ehrling's lemma (see e.g. \cite[Theorem 7.30]{RR}), one has (here $0<s<1/2$ and $B\supset\overline\Omega$ is an open ball)
\begin{align*}
&\langle-\Delta_{\theta}u,u\rangle_{L^{2}(\RE^{n})}=
\|\nabla u\|^{2}_{L^{2}(\Omega_{\-})}+\|\nabla u\|^{2}_{L^{2}(\Omega_{\+})}+\langle\theta[ \gamma_{\Gamma}^{0}]u,[ \gamma_{\Gamma}^{0}]u\rangle_{H^{-s}(\Gamma),H^{s}(\Gamma)}\\
\ge&\|\nabla u\|^{2}_{L^{2}(\Omega_{\-})}+\|\nabla u\|^{2}_{L^{2}(\Omega_{\+})}-2\,{\|}\theta{\|}_{H^{s}(\Gamma),H^{-s}(\Gamma)}\big(\|\gamma^{\-}_{0}u_{\-}\|^{2}_{H^{s}(\Gamma)}+\|\gamma^{\+}_{0}u_{\+}\|^{2}_{H^{s}(\Gamma)}\big)\\
\ge&\|\nabla u\|^{2}_{L^{2}(\Omega_{\-})}+\|\nabla u\|^{2}_{L^{2}(\Omega_{\+})}-c\,{\|}\theta{\|}_{H^{s}(\Gamma),H^{-s}(\Gamma)}\big(\|u_{\-}\|^{2}_{H^{s+1/2}(\Omega_{\-})}+\|u_{\+}\|^{2}_{H^{s+1/2}(\Omega_{\+}\cap B)}\big)\\
\ge&\|\nabla u\|^{2}_{L^{2}(\Omega_{\-})}+\|\nabla u\|^{2}_{L^{2}(\Omega_{\+})}-c\,
{\|}\theta{\|}_{H^{s}(\Gamma),H^{-s}(\Gamma)}
\left(\epsilon\big(\|u_{\-}\|^{2}_{H^{1}(\Omega_{\-})}+\|u_{\+}\|^{2}_{H^{1}(\Omega_{\+}\cap B)}\big)+c_{\epsilon}\|u\|^{2}_{L^{2}(B)}\right)\\
\ge&-\kappa_{\epsilon}\|u\|^{2}_{L^{2}(\RE^{n})}
\end{align*}
and so $\Delta_{\theta}$ is bounded from above. 
\par
By Lemma \ref{lambdateta} and by \eqref{emb}, $\ran(\Lambda_{z}^{\alpha})=H^{1/2}(\Gamma)$ is compactly embedded in $H^{-1/2}(\Gamma)$ . Since $\Gamma$ is bounded, \eqref{HT} hold true. Therefore hypotheses i)-iii) in Theorem \ref{LH} hold.\par 
By \eqref{krein} and \cite[Theorem XIII.13]{RS-IV}, $z\mapsto R^{\theta}_{z}$ has poles (and the coefficients of the Laurent expansion are finite-rank operators) only at $\Sigma_{\theta} $; so,  by \cite[Lemma 1, page 108]{RS-IV}, $\sigma_{disc}(\Delta_\theta)=\Sigma_{\theta} $. Since $\Delta_{\theta}$ is bounded from below, $\Sigma_{\theta}$ is finite.
\par
If $\Omega_{\+}$ is connected, then $\sigma_{p}(\Delta_{\theta})\cap (-\infty,0)=\emptyset$ by the unique continuation principle (see \cite[Remark 3.8]{JST}). 
\par
By taking the limit $\epsilon\downarrow 0$ in the relations (use \eqref{gt0})
$$
(\uno+\theta\Lambda^{N}_{\lambda\pm i\epsilon})(\Lambda^{N}_{\lambda\pm i\epsilon})^{-1}\Lambda^{\theta}_{\lambda\pm i\epsilon}=\uno=(\Lambda^{N}_{\lambda\pm i\epsilon})^{-1}\Lambda^{\theta}_{\lambda\pm i\epsilon}(\uno+\theta\Lambda^{N}_{\lambda\pm i\epsilon})
$$
and
$$
(\uno+\theta\Lambda^{N}_{\lambda\pm i\epsilon})^{-1}=(\Lambda^{N}_{\lambda\pm i\epsilon})^{-1}\Lambda^{\theta}_{\lambda\pm i\epsilon}\,,
$$
one gets the existence of the inverse operator $(\uno+\theta\Lambda^{N,\pm}_{\lambda})^{-1}$
and 
$$
\lim_{\epsilon\downarrow 0}(\uno+\theta\Lambda^{N}_{\lambda\pm i\epsilon})^{-1}=(\uno+\theta\Lambda^{N,\pm}_{\lambda})^{-1}\,.
$$
Thus
$$
\Lambda^{\theta,\pm}_{\lambda}=\lim_{\epsilon\downarrow 0}\Lambda^{N}_{\lambda\pm i\epsilon}(\uno+\theta\Lambda^{N}_{\lambda\pm i\epsilon})^{-1}=
\Lambda^{N,\pm}_{\lambda}(\uno+\theta\Lambda^{N,\pm}_{\lambda})^{-1}=
(Q^{\-}_{\lambda}-Q^{\+}_{\lambda})^{\pm}
(\uno+\theta(Q^{\-}_{\lambda}-Q^{\+}_{\lambda})^{\pm})^{-1}\,.
$$
\end{proof}
\begin{remark} By remark \ref{multi}, we can define $\Delta_{\theta}$ for any real-valued $\theta\in L^{p}(\Gamma)$, $p>2$.
\end{remark}
\begin{remark} In the quantum mechanics oriented literature, the $\delta'$-like boundary conditions are usually represented in terms of a different parameter: let us suppose that $\beta$ is a real-valued function which is a.e. different from zero and such that $\theta:=\beta^{-1}\in L^{p}(\Gamma)$, $p>2$; then one gets the self-adjoint operator $\Delta_{\beta}$ with domain 
$$
\dom(  \Delta_{\beta})=\{ u\in H^{1}_{\Delta}(\RE^{n}): [\gamma_{\Gamma}^{1}]u=0\,,\ \beta\gamma_{\Gamma}^{1}u=[\gamma_{\Gamma}^{0}]u\}\,.
$$
That extends the results contained in \cite[Section 3.2]{BLL}, where $\Delta_{\beta}$ is defined in case $\beta^{-1}\in L^{\infty}(\Gamma)$ and $\Gamma$ is a smooth hypersurface (see also \cite[Section 5.5]{JDE}). In the case $\beta\not=0$ on the measurable set $\Gamma_{\beta}\subsetneq\Gamma$, one can define the corresponding function $\theta$ as $\theta:=\chi_{\beta}\beta^{-1}$, where $\chi_{\beta}$ is the characteristic function of $\Gamma_{\beta}$. Whenever such a function $\theta$ belongs to $L^{p}(\Gamma)$, $p>2$, one gets again a self-adjoint operator $\Delta_{\beta}$, with domain characterized by the boundary conditions 
$$
\dom(  \Delta_{\beta})=\{ u\in H^{1}_{\Delta}(\RE^{n}): [\gamma_{\Gamma}^{1}]u=0\,,\ 
(1-\chi_{\beta})\gamma_{\Gamma}^{1}u=0\,,\ \beta\gamma_{\Gamma}^{1}u=[\gamma_{\Gamma}^{0}]u\}\,.
$$
Operators with such kind of boundary conditions have been constructed (in case $\beta$ and $\theta$ belong to $L^{\infty}(\Gamma)$) in \cite{ER} (see also \cite[Section 6.5]{JDE} for a different construction in the case $\Gamma$ is smooth).  Asymptotic completeness for the scattering couple $(\Delta,\Delta_{\theta})$ provided in Theorem \ref{delta'} extends results on existence and completeness given, in the case the boundary is smooth and $\theta$ is bounded, in \cite{BLL} and \cite{JDE}. The formula for the scattering matrix provided in \eqref{scatt-delta'} extends to Lipschitz hypersurfaces the results given, in the case of a smooth hypersurface and bounded $\theta$, in \cite[Subsections 6.5 and 7.5]{JST}.      
\end{remark}
\end{subsection}
\end{section}
\vskip10pt\noindent

\end{document}